\documentclass[11pt, oneside, reqno]{amsart}

\usepackage[a4paper]{geometry}

\usepackage[english]{babel}

\usepackage[utf8]{inputenc}
\usepackage[T1]{fontenc}

\usepackage{textcomp}

\usepackage{amsfonts, amsmath, amssymb, mathtools, amsthm}
\usepackage[makeroom]{cancel}
\usepackage[none]{hyphenat}
\usepackage{adfbullets}

\usepackage{fancyhdr}
\usepackage{graphicx}
\usepackage[margin=16pt,font=small,labelfont=bf]{caption} %captions aspect customization

\usepackage{enumitem}

\usepackage{parskip}

\newtheorem{theorem}{Theorem}[section]

\newtheorem{proposition}[theorem]{Proposition}
\newtheorem{lemma}[theorem]{Lemma}

\newtheorem{remark}[theorem]{Remark}

\numberwithin{equation}{section}

\DeclarePairedDelimiter{\abs}{\lvert}{\rvert}
\DeclarePairedDelimiter{\norma}{\lVert}{\rVert}
\DeclarePairedDelimiter{\set}{\{}{\}}
\DeclarePairedDelimiter{\angbra}{\langle}{\rangle}
\DeclarePairedDelimiter{\parens}{(}{)}

\DeclarePairedDelimiter{\sqparens}{[}{]}

\newcommand{\Pleft}{\mathrm{P^{dlb}}}
\newcommand{\Pright}{\mathrm{P^{urf}}}

\newcommand{\PplusBC}{\mathbb{P}_{\Lambda}^{+}}
\newcommand{\PminusBC}{\mathbb{P}_{\Lambda}^{-}}

\newcommand{\Pshaken}{\mathrm{P^{sh}}}
\newcommand{\Palt}{\mathrm{P^{alt}}}

\newcommand{\hleft}{h_{x}^{\mathrm{dlb}}}
\newcommand{\hright}{h_{x}^{\mathrm{urf}}}

\newcommand{\altBox}{\Lambda_\mathcal{D}}
\newcommand{\statespace}{\ensuremath{\mathcal{X}}}
\newcommand{\altStatespace}{\ensuremath{\mathcal{X}_\mathcal{D}}}

\usepackage{hyperref}
\hypersetup{
    colorlinks=true,
    linkcolor=blue,
    filecolor=magenta,      
    urlcolor=cyan,
}

\begin{document}
\title{Shaken dynamics on the 3-d cubic lattice}

\author{Benedetto Scoppola}
\address{Benedetto Scoppola, Dipartimento di Matematica
        Università di Roma ``Tor Vergata''    
}

\author{Alessio Troiani}
\address{Alessio Troiani, Dipartimento di Matematica ``Tullio Levi-Civita''
        Università di Padova and Dipartimento di Matematica Università di Roma ``La Sapienza''
}

\author{Matteo Veglianti}
\address{Matteo Veglianti, Dipartimento di Fisica
        Università di Roma ``Tor Vergata''    
}

\begin{abstract}
On the space of $\pm 1$ spin configurations on the 3$d$-square lattice, we consider  the \emph{shaken dynamics}, a parallel Markovian dynamics that can be interpreted in terms of Probabilistic Cellular Automata. The transition probabilities are defined in terms of pair ferromagnetic Ising-type Hamiltonians with nearest neighbor interaction $J$, depending on an additional parameter $q$, measuring the tendency of the system to remain locally in the same state. Odd times and even times
have different transition probabilities. We compute the stationary measure of the shaken dynamics and we investigate its relation with the Gibbs measure for the 3$d$ Ising model. It turns out that the two parameters $J$ and $q$ tune the geometry of the underlying lattice. 
We conjecture the existence of unique line of critical points in $J-q$ plane.
By a judicious use of perturbative methods we delimit the region where such curve must lie and we perform numerical simulation to determine it. Our method allows us to find in a unified way the critical values of $J$ for Ising model with first neighbors interaction, defined on a whole class of lattices, intermediate between the two-dimensional hexagonal and the three-dimensional cubic one, such as, for example, the tetrahedral lattice. Finally we estimate the critical exponents of the magnetic susceptibility and show that our model captures a dimensional transition in the geometry of the system at $q = 0$.
\end{abstract}

\maketitle

\section{Introduction}
Probabilistic Cellular Automata (PCA) are discrete-time Markov chains on a product space $S^{\Lambda}$ (configuration space) whose transition probability is a product measure, i.e. given two generic configurations $\tau=(\tau_1,\dots,\tau_N)$ and $\sigma=(\sigma_1,\dots,\sigma_N)$:
\begin{equation}\label{eq:general_pca_transition_probability}
    P\set*{X_n=\tau|X_{n-1}=\sigma}
    = \prod_{i=1}^N P
        \set*{(X_N)_i=\tau_i|X_{N-1}=\sigma},
\end{equation}
so that for each time $n$, the components of the configuration are independently updated. From a computational point of view, the evolution of a Markov chain of this type is well suited to be simulated on parallel processors.

Recently, a class of PCA has been introduced in order to study nearest neighbors spin systems on lattices and, more generally, spin systems on arbitrary graphs $G=(V,E)$, where the interaction Hamiltonian is given
by
\begin{equation}
    H(\sigma)=-\sum_{e=\{x, y\} \in E} 
                J_{x y} \sigma_{x} \sigma_{y} - 
              2\sum_{x \in V} 
                \lambda_{x} \sigma_{x}
\end{equation}
with both $J_{x y}$ and $\lambda_{x}$ in $\mathbb{R}$, 
and $\sigma \in\{-1,+1\}^{V}$ a \emph{configuration} on $G$. 
In this context, the transitions probability from a configuration
$\sigma$ to a configuration $\tau$ is defined
in terms of a \emph{pair Hamiltonian}  $H(\sigma, \tau)$
and these transitions are such that, 
at each time step, the value of all spins is
simultaneously updated (see: \cite{dss12, ls, dss15, pss, pssboundary}).
In this framework, a new 
parallel dynamics for  Ising-like models on general finite graph
called \emph{shaken dynamics} has been introduced in \cite{shaken2d} 
and has been extensively investigated in the case of the
two dimensional square lattice. 
The distinctive feature of the shaken dynamics is the fact that transitions
between states are obtained through a combination of two 
% \emph{irreversible}
\emph{half steps}. 
In each of these half steps the value of the spin at site $x$
is updated according to a probability distribution depending, 
through a self interaction parameter $q > 0$,
on the value of the spin at site $x$ itself and the values of the spins 
sitting at a suitable subset of the sites adjacent to $x$ in such a way that
all neighbors of $x$ are considered exactly once in the \emph{whole step}.
It is worth noting that
a shaken dynamics on a given graph structure 
can be naturally associated to a dynamics on an \emph{induced} bipartite graph
where the spins in each partition are alternatively updated.
The vertex set of this bipartite graph consists of two copies 
of the original vertex set so that this induced graph can be thought
as to have two \emph{layers}
(see Fig~\ref{fig:layers_structure_induced_graph}). 
\begin{figure}
    \centering
    \includegraphics[width=0.65\textwidth]{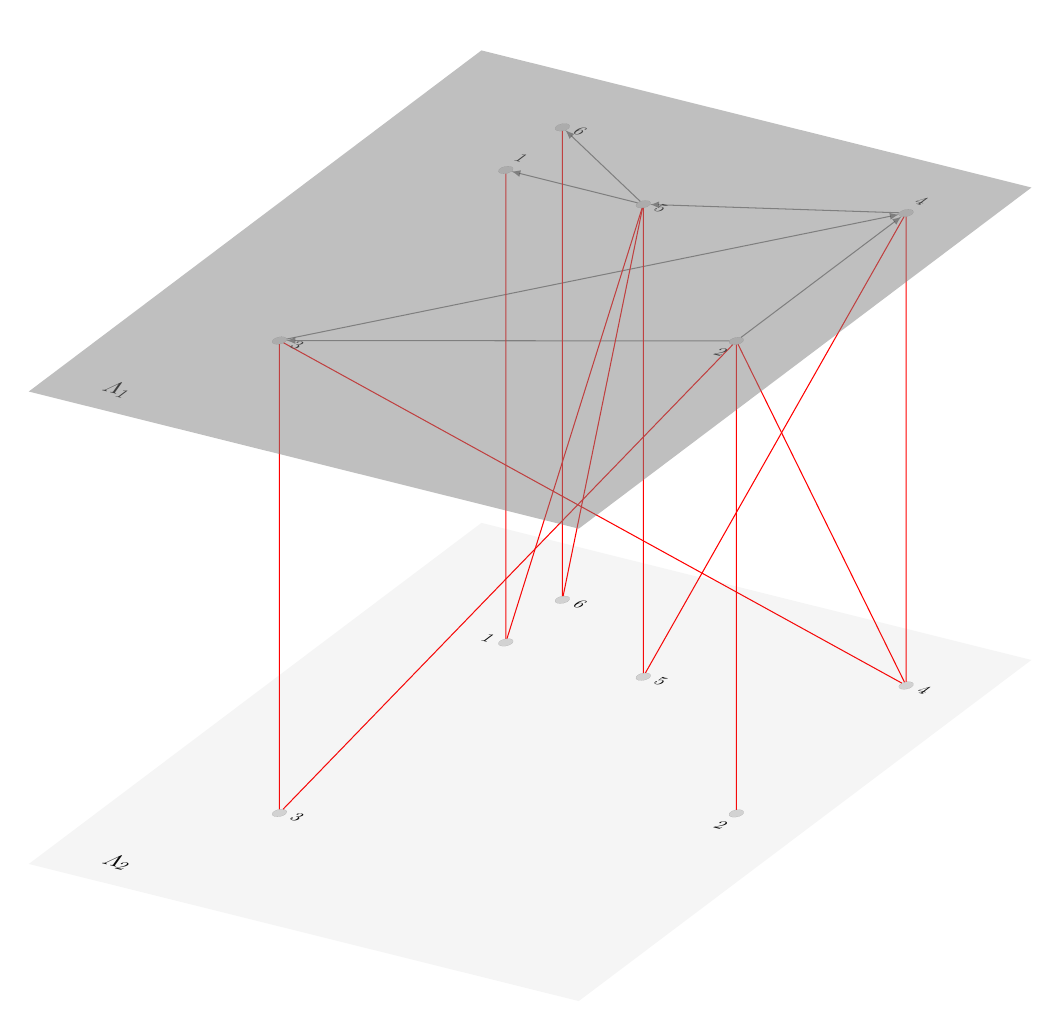}
    \caption{The layers structure of the induced bipartite graph (red edges) associated to the original graph (gray edges). On the induced graph the vertical edges, corresponding to the self interaction, have weight $q$ whereas the slanted edges have weight $J$.}
    \label{fig:layers_structure_induced_graph}
\end{figure}

The sub-configuration on one of the layers  
is the ``current'' configuration of the shaken dynamics whereas the 
sub-configuration on the other layer is the ``intermediate'' configuration 
reached through the first half step. In this view, the shaken dynamics
can be seen as the evolution taking place on one of the layers of the
associated alternate dynamics.
The geometry of the induced bipartite graph where the alternate dynamics
lives varies continuously with $q$. 
For instance, in the case of the shaken dynamics on $\mathbb{Z}^2$ with $J_{xy} = J$ for all pairs of nearest neighbors $\set*{x,y}$, the bipartite graph where the associated alternate dynamics evolves is a non homogeneous hexagonal lattice. For $q = J$ this graph becomes the homogeneous hexagonal lattice;
in the limit $q \to \infty$ the hexagonal lattice ``collapses'' onto 
the square lattice whereas for $q = 0$ the hexagonal lattice becomes a 
collection of independent one dimensional lattices.

In this work we study the shaken dynamics on the $3d$ cubic lattice
with $J_{x,y} = J$ for all $\set*{x,y}$ that are nearest neighbors.
We determine the stationary measure of the shaken dynamics 
and show that if the self interaction $q$
is sufficiently large, then this equilibrium measure 
tends to the Gibbs measure in total variation distance.
% Furthermore
% , in the case of null external magnetic field, 
We argue that the associated alternate dynamics takes place on a suitable
tetrahedral lattice that becomes homogeneous if $J = q$, becomes the
cubic lattice in the limit $q\to\infty$ and reduces to a collection 
of independent $2d$ hexagonal lattices if $q=0$.

It is reasonable to assume that, in the case of null external magnetic field,
there is a critical curve $J_c(q)$ in the $J-q$ plane which separates
the ordered phase from the disordered one. To gain some information
on $J_c(q)$,
we determine two curves in the $J-q$ plane such that 
above the ``upper curve'' the system is in an ordered phase (low temperature regime), whereas below the ``lower curve'' the system is
in a disordered phase (high temperature regime). The 
critical curve must lie in the region delimited by
these two curves.

Further we provide a numerical estimate for $J_c(q)$.
We see that our estimates for
$J_c(0)$, $J_c(J)$ and $J_c(\infty)$ are, respectively,
in good agreement with the critical temperature of the
Ising model on the hexagonal lattice and the numerical
estimates available for the critical temperature of the Ising model
on the tetrahedral and the cubic lattice. This suggests that the 
numerically determined critical curve should be not too far apart
from the ``real'' one.
Moreover we study numerically the critical exponents for the magnetic susceptibility as a function of the self interaction $q$ and provide 
some evidence that the system retains its three dimensional structure
as long as $q > 0$ whereas it becomes two dimensional when $q = 0$.
In other words our model is able to capture the dimensional transition at $q=0$.

In the next section we define the lattice spin model and the shaken dynamics, we describe the first properties of the model, and highlight its relation with the alternate dynamics on the tetrahedral lattice. 
Further we state our main
results. Section 3 is devoted to the proofs. Finally, in section 4 we present our numerical findings concerning the critical curve and discuss the behavior of the critical exponents.

\section{Model description and main results}

\subsection{Definitions and first properties}

Let $\Lambda$ be a $L \times L \times L$ square box in $ \mathbb{Z}^3$ and let $\statespace$ be the set of all possible spin configurations, i.e.
$\statespace=\{\sigma : \sigma=\{-1,1\}^{|\Lambda|}\}$.

We call $B_\Lambda$ be the set of all pairs of nearest neighbors when
periodic boundary conditions are imposed on $\Lambda$.

Let $\sigma, \tau \in \statespace$ be two spin configurations and
define a pair Hamiltonian in the following manner:
\begin{align}
\label{equation1}
H_{\lambda}(\sigma,\tau)  
    & = -\sum_{x\in\Lambda}
            [J\sigma_x(\tau_x^u+\tau_x^r+\tau_x^f)+q\sigma_x\tau_x+\lambda(\sigma_x+\tau_x)]\\
\label{equation2}
    & =-\sum_{x\in\Lambda}         
            [J\tau_x(\sigma_x^d+\sigma_x^l+\sigma_x^b)+q\tau_x\sigma_x+\lambda(\tau_x+\sigma_x)],
\end{align}

where
\begin{itemize}[label=\adfbullet{43}]
  \item $J>0$ represents the ferromagnetic interaction constant,
  \item $q>0$ represents the inertial (or self-interaction) term,
  \item $\lambda>0$ represents the intensity of external magnetic field,
  \item $x^u$ is the site above $x$ 
        at lattice distance $1$ from $x$ itself,
  \item $x^r$ is the site on the right of $x$ 
      %   at lattice distance $1$ from $x$ itself,
  \item $x^f$ is the site in front of $x$
      %   at lattice distance $1$ from $x$ itself,
  \item $x^d$ is the site below (down) $x$ 
      %   at lattice distance $1$ from $x$ itself,
  \item $x^l$ is the site on the left of $x$
      %   at lattice distance $1$ from $x$ itself,
  \item $x^b$ is the site behind $x$
      %   at lattice distance $1$ from $x$ itself,
  \item $\sigma_x$ (resp. $\tau_x$) 
        is the spin at site $x$ 
        in configuration $\sigma$ (resp. $\tau$)
  \item $\sigma_x^d$ is the spin at site $x^d$
        in configuration $\sigma$ ($\sigma_x^l$, 
        $\sigma_x^b$, $\tau_x^u$, \ldots
        are defined likewise)
\end{itemize}
See Fig.~\ref{cubic}. 

\begin{figure}
\centering
\includegraphics[scale=0.5]{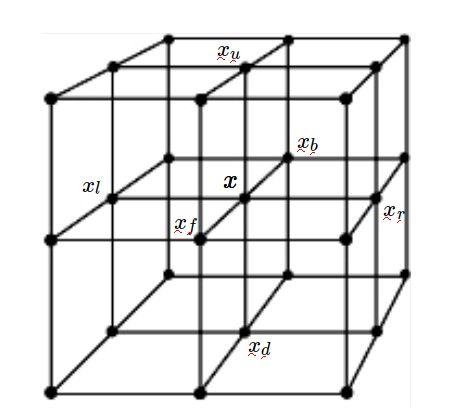}
\caption{cubic lattice.}
\label{cubic}
\end{figure}

It is straightforward to check that
Hamiltonian \eqref{equation1} is linked tightly
to the standard Ising one on the standard cubic lattice. 
In particular the 
following proposition holds.
\begin{proposition}
\label{proposition1}
\begin{equation}
\label{equation3}
H_{\lambda}(\sigma,\sigma)=H_{2\lambda}(\sigma)-q|\Lambda|,
\end{equation}\\
with
\begin{equation}
\label{equation4}
H_{2\lambda}(\sigma)=-\sum_{(x,y)\in B_{\Lambda}} J\sigma_x\sigma_y-2\lambda\sum_{x\in\Lambda} \sigma_x,
\end{equation}
the standard Ising Hamiltonian with external magnetic field twice that of Hamiltonian (\ref{equation1}).
\end{proposition}
\begin{proof}

Follows immediately from (\ref{equation1}) by setting $\sigma=\tau$, see \cite{shaken2d}.
\end{proof}

In the same spirit of \cite{shaken2d} we want to define a \emph{shaken dynamics}
on $\statespace$. 
To this end, consider a Markov Chain that updates
the spin configuration with transitions probability $\Pleft$ at odd
times and $\Pright$ at even times where:

\begin{equation}
\label{equation7}
\Pleft(\sigma,\sigma')=\frac{e^{-H_{\lambda}(\sigma,\sigma')}}{\overrightarrow{\rm Z_\sigma}} \ \ \ \text{and} \ \ \ \Pright(\sigma,\sigma')=\frac{e^{-H_{\lambda}(\sigma',\sigma)}}{\overleftarrow{\rm Z_\sigma}},
\end{equation}

with 
$\overrightarrow{\rm Z_\sigma}=\sum_{\sigma'\in\chi} e^{-H_{\lambda}(\sigma,\sigma')} \ \ \ \text{and} \ \ \ \overleftarrow{\rm Z_\sigma}=\sum_{\sigma'\in\chi} e^{-H_{\lambda}(\sigma',\sigma)}$
normalizing constants.

Then, the shaken dynamics is defined through the composition of
an ``odd'' and an ``even'' step. More precisely:
\begin{align}
\label{equation13}
\Pshaken(\sigma,\tau)=\sum_{\sigma'\in\chi} \Pleft(\sigma,\sigma')\Pright(\sigma',\tau)=\sum_{\sigma'\in\chi} \frac{e^{-H_{\lambda}(\sigma,\sigma')}}{\overrightarrow{\rm Z_\sigma}}\frac{e^{-H_{\lambda}(\tau,\sigma')}}{\overleftarrow{\rm Z_\sigma'}}.
\end{align}

Though, strictly speaking, the shaken dynamics \eqref{equation13} is
not a PCA in the sense of \eqref{eq:general_pca_transition_probability},
it is the composition of two steps each having a 
factorized transition probability. Indeed:

\begin{align}
\Pleft(\sigma,\sigma')=
    \prod_{x\in\Lambda}\frac{e^{\hleft(\sigma)\sigma_x'}}{2\cosh \hleft(\sigma)}, \quad
\Pright(\sigma,\sigma')=
    \prod_{x\in\Lambda}\frac{e^{\hright(\sigma)\sigma_x'}}{2\cosh \hright(\sigma)}
\end{align}
where
\begin{align}
    \hleft(\sigma)=
        J(\sigma_x^d+\sigma_x^l+\sigma_x^b)+q\sigma_x-\lambda,
    \quad    
    \hright(\sigma)=
        J(\sigma_x^u+\sigma_x^r+\sigma_x^f)+q\sigma_x+\lambda
\end{align}
are the local fields felt at site $x$ at, respectively, the odd and the
even ``half steps''.

Observe that the Hamiltonian \eqref{equation1} is not symmetric:
$
    H_{\lambda}(\sigma,\tau)\neq H_{\lambda}(\tau,\sigma).
$
This implies that a dynamics evolving solely according to
$\Pleft$ or $\Pright$ is not reversible. However, when the shaken 
dynamics \eqref{equation13} is considered, then the following result holds:

\begin{proposition}\label{proposition4}
The shaken dynamics $\Pshaken(\sigma,\tau)$ is reversible with respect to the measure
$
\pi_\Lambda(\sigma)=\frac{\overrightarrow{\rm Z_\sigma}}{Z}
$,
with $Z$ a normalizing constant.
% which is, therefore, its stationary measure.
\end{proposition}
\begin{proof}
The detailed balance condition is readily established, indeed:
\begin{align}
\begin{aligned}
\overrightarrow{\rm Z_\sigma}\Pshaken(\sigma,\tau)
    & = \overrightarrow{\rm Z_\sigma}\sum_{\sigma'\in\chi} \frac{e^{-H_{\lambda}(\sigma,\sigma')}}{\overrightarrow{\rm Z_\sigma}}\frac{e^{-H_{\lambda}(\tau,\sigma')}}{\overleftarrow{\rm Z_\sigma'}}
    = \sum_{\sigma'\in\chi} \frac{e^{-[H_{\lambda}(\sigma,\sigma')+H_{\lambda}(\tau,\sigma')]}}{\overleftarrow{\rm Z_\sigma'}} \\
    & = \overrightarrow{\rm Z_\tau}\sum_{\sigma'\in\chi} \frac{e^{-H_{\lambda}(\tau,\sigma')}}{\overrightarrow{\rm Z_\tau}}\frac{e^{-H_{\lambda}(\sigma,\sigma')}}{\overleftarrow{\rm Z_\sigma'}}
    = \overrightarrow{\rm Z_\tau}\Pshaken(\tau,\sigma)
\end{aligned}
\end{align}
\end{proof}

\begin{figure}
\centering
\includegraphics[scale=0.6]{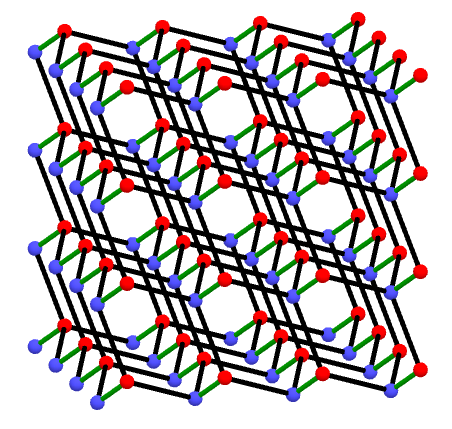}
\caption{The graph $\altBox$.
Blue and red dots represent respectively the sites in the two three-dimensional lattices $V_1$ and $V_2$. The black segments correspond to the interaction governed by the parameter $ J $, the green segments correspond to the self-interaction, governed by the parameter q.}
\label{reticolo1}
\end{figure}

\subsection{Alternate dynamics}
Let $V_1$ and $V_2$ be two copies of $V$ the vertex set of $\Lambda$
and let $\altBox$ be a finite graph with vertex set given by
$V_1 \cup V_2$. 
For each site $x$ in $\Lambda$, $x_1, x_2$ are the two copies of
$x$ in, respectively, $V_1$ and $V_2$ and are called corresponding sites.
Consider a spin configuration $\sigma_1$ on $V_1$ and a spin
configuration $\sigma_2$ on $V_2$. Then, the Hamiltonian
$H(\sigma_1, \sigma_2)$ defines the set of edges on $\altBox$.
In particular each pair $x_1, x_2$ of corresponding sites is connected
by an edge with weight $q$. Moreover, each $x_1 \in V_1$ has there
additional edges, $\{x_1, x_2^u\}$, $\{x_1, x_2^r\}$, $\{x_1, x_2^f\}$,
with weight $J$. 

Note that the graph $\altBox$ is bipartite by construction and each edge
has one endpoint in $V_1$ and one in $V_2$. 
A graphical representation of $\altBox$ is given in Fig.~\ref{reticolo1}.

The parameter $q$ determines the geometry of the lattice $\altBox$. Indeed,
thinking to the edge weight as to be proportional to the inverse of the
geometrical distance between the vertices we have:
\begin{itemize}[label=\adfbullet{43}]
\item the limit $q\to 0$ correspond to erasing the $q-$edges obtaining, from the lattice $\altBox$, ``independent'' copies of the two-dimensional honeycomb lattice;
\item when $J=q$, the $q-$edges and the $J-$edges become of the same length, so the lattice $\altBox$ becomes a tetrahedral lattice, that we can imagine like a diamond structure;
\item the limit $q\to\infty$ correspond to identify the two vertices linked by the $q-$edge, in this case the lattice $\altBox$ degenerates into a simple cubic lattice.
\end{itemize}

Consider a dynamics (in the remainder referred to as
\emph{alternate dynamics}) that, alternatively, at each step updates 
all spins in one of the two layers. Then the shaken dynamics can be
seen as the projection of the alternate dynamics onto one of the layers.

To make this statement precise, let
$\vec{\sigma} = (\sigma_1,\sigma_2), \; 
\vec{\tau} = (\tau_1, \tau_2) \in \altStatespace$. 
Then, the alternate dynamics on $\altStatespace$ is defined by
the transition probabilities
\begin{equation}
\label{equation18}
    \Palt(\vec{\sigma},\vec{\tau})
        = \Pleft(\sigma_1,\tau_2)\Pright(\tau_2,\tau_1)
        = \frac{e^{-H_{\lambda}(\sigma_1,\tau_2)}}
               {\overrightarrow{\rm Z_ {\sigma_1}}}
          \frac{e^{-H_{\lambda}(\tau_1,\tau_2)}}
               {\overleftarrow{\rm Z_{\tau_2}}}.
\end{equation}
and the transition probabilities of the shaken dynamics can be
written in the form
\begin{equation}
\label{equation17}
\Pshaken(\sigma_1,\tau_1)
    =\sum_{\tau_2\in\statespace} \Palt ((\sigma_1,\cdot), (\tau_1,\tau_2)),
\end{equation}
Let $ Z_{\altBox} = \sum_{\sigma, \tau} e^{-H(\sigma, \tau)}$.
As far as the stationary measure of the alternate dynamics is concerned
we have the following:
\begin{proposition}
The alternate dynamics defined on $\altStatespace$ with transition probability $\Palt(\vec{\sigma},\vec{\tau})$ has the following stationary measure:
\begin{equation}
    \pi_2(\sigma,\tau)=\frac{1}{Z_{\altBox}} e^{-H_{\lambda}(\sigma,\tau)}.
\end{equation}
Moreover, in general, this dynamics is irreversible.
\end{proposition}
\begin{proof}
\begin{align}
\begin{aligned}
    \sum_{\sigma_1,\sigma_2}
        \pi_2 (\sigma_1,\sigma_2) \Palt(\vec{\sigma},\vec{\tau}) 
    & = \sum_{\sigma_1,\sigma_2} 
        \frac{1}
             {Z} 
        e^{-H(\sigma_1,\sigma_2)}
        \frac{e^{-H_{\lambda}(\sigma_1,\tau_2)}}
             {\overrightarrow{\rm Z_{\sigma_1}}}
        \frac{e^{-H_{\lambda}(\tau_1,\tau_2)}}
             {\overleftarrow{\rm Z_{\tau_2}}} \\
    & = \sum_{\tau_1,\tau_2} \frac{1}{Z} e^{-H_{\lambda}(\sigma_1,\sigma_2)}
    = \pi_2 (\tau_1,\tau_2).
\end{aligned}
\end{align}
However, in general,
\begin{equation}
    \pi_2 (\sigma_1,\sigma_2) \Palt(\vec{\sigma},\vec{\tau}) 
    \neq 
    \pi_2 (\tau_1,\tau_2) \Palt(\vec{\tau}, \vec{\sigma}).
\end{equation}
Consider, for instance, the transition from
$\vec{\sigma} = (\underline{+1}, \underline{+1})$ to
$\vec{\tau} = (\underline{+1}, \underline{-1})$.
\end{proof}
\begin{remark}
    The stationary measure of the shaken dynamics is the marginal
    of the stationary of the alternate dynamics, that is
    \begin{equation}
        \pi_\Lambda(\sigma) = \sum_{\tau \in \statespace} \pi_2(\sigma, \tau). 
    \end{equation}
\end{remark}

\subsection{Results}
% In this Section we state our main results.
In the case of null external magnetic field, we
identify two regions of analiticity of the partition function
in the thermodynamic limit.
These two regions correspond, respectively, to a low temperature
and a high temperature regime for the system.
The estimation of the critical curve expected to lie outside of these two regions is provided below (see Section~\ref{sec:numerical_simulations}).
Moreover, for large $q$ we establish a link between
the equilibrium measure of the shaken dynamics and
the Gibbs measure for the Ising model defined on
the standard cubic lattice.
This result, presented in Theorem~\ref{theorem1},
gives a quantitative support to the statement that,
if $q$ is sufficiently large, the lattice where the
alternate dynamics takes place, tends to the
simple cubic one.

We define, for the model introduced previously, the Gibbs measure:
$
\pi_{\Lambda}^G(\sigma)=\frac{e^{-H_{2\lambda}(\sigma)}}{Z^G},
$
\\with
$
Z^G=\sum_{\sigma \in \chi} e^{-H_{2\lambda}(\sigma)},
$
where $H_{2\lambda}(\sigma)$ is given by (\ref{equation4}).

The stationary measure of the Markov chain defined above, that is $\pi_{\Lambda}(\sigma)$, 
is linked to the Gibbs measure $\pi_{\Lambda}^G(\sigma)$ by the following result:

\begin{theorem}\label{theorem1}
    Let $\delta=e^{-2q}$.
    If $\lim_{\abs{\Lambda} \to \infty} \delta^2 \abs{\Lambda}=0$,
    then there exist a $\bar{J}$ and $\tilde{J}$,
    with $\bar{J} > \tilde{J}$, such that
    for $J> \bar J$ and $J<\tilde{J}$ we have:
    \begin{equation}\label{thm1}
    \lim_{\abs{\Lambda} \to \infty} \norma{\pi_{\Lambda}-\pi_{\Lambda}^G}_{TV}=0.\footnotemark
    \footnotetext{$\norma{\cdot}_{TV}$ denotes the total variation distance:
    $\norma{\pi_{\Lambda}-\pi_{\Lambda}^G}_{TV}=\frac{1}{2}\sum_{\sigma \in \chi}   \abs{\pi_{\Lambda}(\sigma)-\pi_{\Lambda}^G(\sigma)}$}
    \end{equation}
\end{theorem}

\textbf{Remark}: In this paper we are considering periodic boundary condition on $\Lambda$. However Theorem~\ref{theorem1} holds also if ``plus'' or ``minus'' boundary conditions are imposed at the exterior of the box. Indeed, if we consider a set $B$ of fixed spins in $\Lambda$, such that $|B|/|\Lambda| \to 0$ as $|\Lambda| \to \infty$,
the proof of the theorem requires only minor modifications (see below). Using this approach,
any kind of external boundary conditions can easily be mimicked by fixing the spins on three mutually orthogonal planes.

To identify the low temperature regime, we look at
the magnetization of the origin and determine parameters such
that this magnetization
is positively correlated with the magnetization at the boundary even in the
thermodynamic limit.

Let $\pi^{+}_{\Lambda}(\sigma)$ (resp. $\pi^{-}_{\Lambda}(\sigma)$)  be the equilibrium measure of the
shaken dynamics when $+$ (resp. $-$) boundary conditions 
are taken into account and let 
$\angbra{\sigma_{0}}^{+}_{\Lambda}$ 
(resp. $\angbra{\sigma_{0}}^{-}_{\Lambda}$) be the expected value of $\sigma_{0}$
with respect to this probability measure, that is
$\angbra{\sigma_{0}}^{+}_{\Lambda} = \sum_{\sigma} \sigma_{0} \pi^{+}_{\Lambda}(\sigma)$
and
$\angbra{\sigma_{0}}^{-}_{\Lambda} = \sum_{\sigma} \sigma_{0} \pi^{-}_{\Lambda}(\sigma)$. Then, for $\lambda = 0$:
\begin{theorem}\label{thm:low_temp_regime}
    In the thermodynamic limit the mean magnetization 
    of the origin depends on the boundary conditions, that is
    \begin{equation}
        \lim_{\abs{\Lambda} \to \infty} \angbra{\sigma_{0}}^{+}_{\Lambda}
        \neq
        \lim_{\abs{\Lambda} \to \infty} \angbra{\sigma_{0}}^{-}_{\Lambda}
    \end{equation}
    if $J$ and $q$ are sufficiently large. The explicit
    description of the low temperature
    region is given in \eqref{condtrans}.
\end{theorem}

Conversely, when the system is at high temperature, 
it is possible to bound
the analyticity region of the free-energy density
by considering a suitable (high temperature)
expansion. More precisely, for $\lambda = 0$:
\begin{theorem}\label{thm:high_temp_regime}
    The region in the $q-J$ plane where the free-energy density 
    \begin{equation}\label{sigmahightemp}
        f_{\altBox}(J,q) = \frac{1}{\abs{\altBox}} \ln
        Z_{\altBox}(J,q)
    \end{equation} 
    is analytic in the thermodynamic limit contains a well defined region specfied by \eqref{FP13} below.  
\end{theorem}

\section{Proofs of the main results}
\subsection{Proof of theorem \ref{theorem1}}

We start proving the first part of the theorem:
if \linebreak 
\mbox{$
    \lim_{\abs{\Lambda} \to \infty} \delta^2 \abs{\Lambda}=0
$}, 
then there exists a 
$\bar J$ such that, for $J> \bar J$ we have:
\linebreak
$
\lim_{\abs{\Lambda} \to \infty} \norma{\pi_{\Lambda}-\pi_{\Lambda}^G}_{TV}=0.
$

To this end, we need some preliminaries lemmas.
\begin{lemma}
\label{lemma1}
\begin{equation}
\label{equation23}
\overrightarrow{Z_{\sigma}}=e^{q\abs{\Lambda}}e^{-H_{2\lambda}(\sigma)}\prod_{x \in \Lambda} \left(1+\delta e^{-2g_x^{dlb}(\sigma)\sigma_x-2\lambda\sigma_x}\right),
\end{equation}
where:
\begin{equation}
    g_x^{dlb}(\sigma)=J(\sigma_x^d+\sigma_x^l+\sigma_x^b).    
\end{equation}

\end{lemma}
\begin{proof}
It follows using the same steps of \cite[Equation~25]{shaken2d}
\end{proof}

In order to compare
the stationary measure 
$\pi_\Lambda$ with the Gibbs measure
$\pi_\Lambda^G$, it is convenient to
rewrite the previous
expression 
for $\overrightarrow{Z_{\sigma}}$,
in terms of the \emph{Gibbs weight}
$\omega^G(\sigma)=e^{-H_{2\lambda}(\sigma)}$ of configuration $\sigma$.

Write
\begin{equation}
\label{omegasigma}
\omega(\sigma)=e^{-H_{2\lambda}(\sigma)}f(\sigma)=\omega^G(\sigma)f(\sigma),
\end{equation}
with
\begin{equation}
\label{equation27}
    f(\sigma)=\prod_{x \in \Lambda} 
        \sqparens*{
            1+\delta e^{-2g_x^{dlb}(\sigma)\sigma_x-2\lambda \sigma_x}
        }.
\end{equation}

Then $\overrightarrow{Z_\sigma}$ can be written as
\begin{equation}
    \overrightarrow{\rm Z_\sigma}=\omega^G(\sigma)e^{q|\Lambda|}f(\sigma)
             =\omega(\sigma)e^{q|\Lambda|}.
\end{equation}

Recalling the definition of the Gibbs measure: 
\begin{equation}
    \pi_{\Lambda}^G(\sigma)=\frac{e^{-H_{2\lambda}(\sigma)}}{Z^G}=\frac{\omega^G(\sigma)}{Z^G}
\end{equation}
then $\pi_\Lambda$ can be written as:
\begin{equation}
    \pi_{\Lambda}(\sigma)
    = \frac{\overrightarrow{\rm Z_\sigma}}{Z}
    =\frac{\frac{\omega^G(\sigma)}{Z^G}f(\sigma)}
          {\sum_{\sigma \in\statespace}
                \frac{\omega^G(\sigma)}{Z^G}f(\sigma)
          }
    =\frac{\pi^G_{\Lambda}(\sigma)f(\sigma)}{\pi^G_{\Lambda}(f)},
\end{equation}
where:
\begin{equation}
\label{pigf}
    \pi^G_{\Lambda}(f)=\sum_{\sigma\in\statespace}
                        \pi^G_{\Lambda}(\sigma)f(\sigma).
\end{equation}

With this notation, the following lemma provides a bound for the \emph{difference} between the two measures on $\statespace$.
\begin{lemma}
\label{lemma2}
\begin{equation}
\norma{\pi_{\Lambda}-\pi_{\Lambda}^G}_{TV} \leq [\Delta(\delta)]^{\frac{1}{2}},
\end{equation}
with:
\begin{equation}
\Delta(\delta)=\frac{\pi^G_{\Lambda}(f^2)}{(\pi^G_{\Lambda}(f))^2}-1.
\end{equation}
\end{lemma}
\begin{proof}
See \cite[Proof of Theorem~2.5]{shaken2d}
\end{proof}

Hence, to prove the first part of theorem (\ref{theorem1}), we need to show that:
\begin{equation}
\label{equation35}
\Delta(\delta)=\mathcal{O}(\delta^2\abs{\Lambda}).
\end{equation}

Writing:
\begin{equation}\label{eq.Delta}
    \Delta(\delta)=\frac{\pi^G_{\Lambda}(f^2)}{\pi^G_{\Lambda}(f)^2}-1=e^{\log{[\pi^G_{\Lambda}(f^2)]}-2\log{[\pi^G_{\Lambda}(f)]}}-1,
\end{equation}

then we need to show that $\exists \, \bar{J}$ such that, for 
$J > \bar{J}$, we have:
\begin{enumerate}[label=\emph{\alph*})]
    \item
        The functions
        \begin{equation}\label{condition1}
            \frac{\log{[\pi^G_{\Lambda}(f^2)]}}{\abs{\Lambda}}
            \text{\, and \,}
            \frac{\log{[\pi^G_{\Lambda}(f)]}}{\abs{\Lambda}}
        \end{equation}
        are both analytic for $\abs{\delta}<\delta_J$ for
        a suitable $\delta_J$ depending on $J$
    \item
        \begin{equation}\label{condition2}
        \frac{\log{[\pi^G_{\Lambda}(f^2)]}}{\abs{\Lambda}}-2\frac{\log{[\pi^G_{\Lambda}(f)]}}{\abs{\Lambda}}=\mathcal{O}(\delta^2).
        \end{equation}
\end{enumerate}

Indeed, if \eqref{condition1} holds, the Taylor expansion of the exponential of
the r.h.s. of \eqref{eq.Delta} is well defined.
\eqref{condition2} follows by noting that the first order
terms of the Taylor expansions are zero.

To prove both claims, 
it is convenient to partition the sites of 
the finite cubic lattice $\Lambda$ according to
the value of their spin and the 
sum of the spins at the 
downwards, leftwards and backwards neighboring sites.
To this purpose we define the sets
\begin{itemize}[leftmargin=5mm, label=\adfbullet{43}]
\item
    $ N_3^{-} 
        = \{ x \in \Lambda : \sigma_x=-1 \land \sigma_x^d+\sigma_x^l+\sigma_x^b=-3 \}
        = \{ x \in \Lambda : \sigma_x=-1 \land g_x^{dlb}(\sigma)=-3J \}$,
\item
    $ N_2^{-}
        = \{ x \in \Lambda : \sigma_x=-1 \land \sigma_x^d+\sigma_x^l+\sigma_x^b=-1 \} 
        = \{ x \in \Lambda : \sigma_x=-1 \land g_x^{dlb}(\sigma)=-1J \}, $
\item
    $ N_1^{-}
        = \{ x \in \Lambda : \sigma_x=-1 \land \sigma_x^d+\sigma_x^l+\sigma_x^b=+1 \} 
        = \{ x \in \Lambda : \sigma_x=-1 \land g_x^{dlb}(\sigma)=+1J \}, $
\item
    $ N_0^{-}
        = \{ x \in \Lambda : \sigma_x=-1 \land \sigma_x^d+\sigma_x^l+\sigma_x^b=+3 \} 
        = \{ x \in \Lambda : \sigma_x=-1 \land g_x^{dlb}(\sigma)=+3J \}, $
\item
    $ N_3^{+}
        = \{ x \in \Lambda : \sigma_x=+1 \land \sigma_x^d+\sigma_x^l+\sigma_x^b=+3 \} 
        = \{ x \in \Lambda : \sigma_x=+1 \land g_x^{dlb}(\sigma)=+3J \}, $
\item
    $ N_2^{+}
        = \{ x \in \Lambda : \sigma_x=+1 \land \sigma_x^d+\sigma_x^l+\sigma_x^b=+1 \} 
        = \{ x \in \Lambda : \sigma_x=+1 \land g_x^{dlb}(\sigma)=+1J \}, $
\item
    $ N_1^{+}
        = \{ x \in \Lambda : \sigma_x=+1 \land \sigma_x^d+\sigma_x^l+\sigma_x^b=-1 \} 
        = \{ x \in \Lambda : \sigma_x=+1 \land g_x^{dlb}(\sigma)=-1J \}, $
\item
    $ N_0^{+}
        = \{ x \in \Lambda : \sigma_x=+1 \land \sigma_x^d+\sigma_x^l+\sigma_x^b=-3 \} 
        = \{ x \in \Lambda : \sigma_x=+1 \land g_x^{dlb}(\sigma)=-3J \}. $
\end{itemize}
Checking that 
$
\Lambda = N_3^{-} \cup N_2^{-} \cup N_1^{-} \cup N_0^{-} \cup N_3^{+} \cup N_2^{+} \cup N_1^{+} \cup N_0^{+}$
is straightforward.
\\
Then, arguing as in \cite{shaken2d}, it is possible to rewrite $f(\sigma)$ in this way:
\begin{equation}
\label{eqf}
f(\sigma)=(1+\delta e^{-6J-2\lambda})^{\abs{\Lambda}} \tilde{\xi}(\sigma,\lambda),
\end{equation}
\\
with:
\begin{align}
\label{xi}
\begin{aligned}
    \tilde{\xi}(\sigma,\lambda) 
        & = \Bigl( \frac{1+\delta e^{-6J+2\lambda}}{1+\delta e^{-6J-2\lambda}} \Bigr)^{\abs{N_3^-}} 
        \Bigl( \frac{1+\delta e^{-2J+2\lambda}}{1+\delta e^{-6J-2\lambda}} \Bigr)^{\abs{N_2^-}} 
        \Bigl( \frac{1+\delta e^{2J+2\lambda}}{1+\delta e^{-6J-2\lambda}} \Bigr)^{\abs{N_1^-}}
        \\
        & \times \Bigl( \frac{1+\delta e^{6J+2\lambda}}{1+\delta e^{-6J-2\lambda}} \Bigr)^{\abs{N_0^-}}
        \Bigl( \frac{1+\delta e^{-2J-2\lambda}}{1+\delta e^{-6J-2\lambda}} \Bigr)^{\abs{N_2^+}}
        \\
        & \times \Bigl( \frac{1+\delta e^{2J-2\lambda}}{1+\delta e^{-6J-2\lambda}} \Bigr)^{\abs{N_1^+}} 
        \Bigl( \frac{1+\delta e^{6J-2\lambda}}{1+\delta e^{-6J-2\lambda}} \Bigr)^{\abs{N_0^+}}.
\end{aligned}    
\end{align}

To bound $\tilde{\xi}$, we rewrite $H_{2\lambda}(\sigma)$ in terms of \emph{$3d$-Peierls contours}
defined in the following way: 
for each pair nearest neighboring sites $x$ and $y$
such that
$\sigma_x\sigma_y = - 1$, we build a square unitary plate that is orthogonal to the segment between $\sigma_x$ and $\sigma_y$ and passing through the midpoint of this segment.
In this way, starting from a spin configuration $\sigma$, we can introduce a family of closed polyhedra (or $3d$-Peierls contours configuration) $\Gamma(\sigma)=\{\gamma_1,\dots,\gamma_N\}$ separating the regions with spin $ + $ 1 from those with spin $ -1 $.\footnote{The correspondence between $\sigma$ and $\Gamma(\sigma)$ is one to two for periodic boundary conditions and one to one if at least one spin is fixed, this includes the case of ``plus'' or ``minus'' boundary conditions.}

Denote by $B^-$ the total number of $-1$ bonds
in $B_\Lambda$, that is the total number of 
edges with spins of opposite sign on its endpoints
and by $B_{TOT}$ the total number of bonds
in $B_\Lambda$.

Denoting by 
\begin{equation}
\abs{\Gamma(\sigma)}=\sum_{\gamma_i \in \Gamma} \abs{\gamma_i},
\end{equation}

the total number of plates of the contours configuration, 
we clearly have
\begin{equation}
\label{totalcontours}
\abs{\Gamma(\sigma)}=B^-.
\end{equation}

Performing simple algebraic calculations, we can rewrite $H_{2\lambda}(\sigma)$:
\begin{align}
H_{2\lambda}(\sigma)=-J (-2B^-+B_{TOT})-2\lambda \abs{\Lambda} + 4\lambda \abs{V_{-}(\sigma)},
\end{align}
where $V_{-}(\sigma)$ is the number of sites with negative spin.\\
We have: $B_{TOT}=3|\Lambda|$.
Moreover, using \eqref{totalcontours}:
\begin{equation}
\label{e^HPeierls}
e^{-H_{2\lambda}(\sigma)}=e^{(3J+2\lambda)\abs{\Lambda}-2J\abs{\Gamma(\sigma)}-4\lambda\abs{V_-(\sigma)}}.
\end{equation}
Using this result, we can write, for $k = 1,2$,:
\begin{align}\label{piPeierls}
\begin{aligned}
    \pi_{\Lambda}^G(f^k)
        & =%\frac{1}{Z_G} 
        \frac{e^{(3J\abs{\Lambda}+2\lambda)}}{Z_G} \left( 1+\delta e^{-6J+2\lambda}\right)^{k\abs{\Lambda}} \sum_{\sigma} \left[ e^{-2J\abs{\Gamma(\sigma)}}e^{-4\lambda\abs{V_-(\sigma)}} \tilde{\xi}^k(\sigma,\lambda)\right]\\
        & =%\frac{1}{Z_G} 
        \frac{e^{(3J\abs{\Lambda}+2\lambda)}}{Z_G} \left( 1+\delta e^{-6J+2\lambda}\right)^{k\abs{\Lambda}} \sum_{\sigma} \left[ e^{-2J\abs{\Gamma(\sigma)}}(e^{-2\lambda\abs{V_-(\sigma)}})^{2-k} \xi^k(\sigma,\lambda)\right],
\end{aligned}
\end{align}
with:
\begin{align}
\begin{aligned}
    \xi^k(\sigma, \lambda) 
    & = \parens*{
            \frac{e^{-2\lambda}(1+\delta e^{-6J+2\lambda})}{1+\delta e^{-6J-2\lambda}}}^{k\abs{N_3^-}} 
        \parens*{
            \frac{e^{-2\lambda}(1+\delta e^{-2J+2\lambda})}{1+\delta e^{-6J-2\lambda}}}^{k\abs{N_2^-}}
            \\
    & \times
        \parens*{
            \frac{e^{-2\lambda}(1+\delta e^{2J+2\lambda})}{1+\delta e^{-6J-2\lambda}} }^{k\abs{N_1^-}}
        \parens*{
            \frac{e^{-2\lambda}(1+\delta e^{6J+2\lambda})}{1+\delta e^{-6J-2\lambda}} }^{k\abs{N_0^-}}
            \\
    & \times
        \parens*{
            \frac{1+\delta e^{-2J-2\lambda}}{1+\delta e^{-6J-2\lambda}} }^{k\abs{N_2^+}} 
        \parens*{
            \frac{1+\delta e^{2J-2\lambda}}{1+\delta e^{-6J-2\lambda}} }^{k\abs{N_1^+}} 
        \parens*{            \frac{1+\delta e^{6J-2\lambda}}{1+\delta e^{-6J-2\lambda}} }^{k\abs{N_0^+}}.
\end{aligned}
\end{align}

It is now straightforward to prove the next technical lemma: 
\begin{lemma}
\label{lemma3}
\begin{equation}
\xi^k(\sigma, \lambda) \leq \xi^k(\sigma, 0), \text{ \ for both \ } k=1 \text{ \ and \ } k=2.
\end{equation}
\end{lemma}
\begin{proof}
A simple algebraic calculation shows that each factor of $\xi^k(\sigma, \lambda)$ is less or equal to the respectively factor of $\xi^k(\sigma, 0)$ for both $k=1$ and $k=2$.\\
\end{proof}

As a consequence of the previous lemma:
\begin{align}\label{eq:from_conf_to_contours}
\begin{aligned}
    \sum_{\sigma} 
        \left[
            e^{-2J\abs{\Gamma(\sigma)}}(e^{-2\lambda\abs{V_-(\sigma)}})^{2-k} \xi^k(\sigma,\lambda)
        \right] 
    & \leq \sum_{\sigma} 
            \left[
                e^{-2J\abs{\Gamma(\sigma)}} \xi^k(\sigma,0)
            \right]\\
    & = 2\sum_{\Gamma}
        \left[
            e^{-2J\abs{\Gamma}} \xi^k(\Gamma,0)
        \right].
\end{aligned}
\end{align}
And hence:
\begin{equation}\label{pigamma}
\pi_{\Lambda}^G(f^k) \leq \frac{2}{Z_G} e^{(3J+2\lambda)\abs{\Lambda}} \left( 1+\delta e^{-6J-2\lambda}\right)^{k\abs{\Lambda}} \sum_{\Gamma} \left[ e^{-2J\abs{\Gamma}} \xi^k(\Gamma,0)\right],
\end{equation}
with:
\begin{equation}
\label{xigamma}
\xi^k(\Gamma,0)=\Bigl( \frac{1+\delta e^{-2J}}{1+\delta e^{-6J}} \Bigr)^{k(\abs{N_2^-}+\abs{N_2^+})} \Bigl( \frac{1+\delta e^{2J}}{1+\delta e^{-6J}} \Bigr)^{k(\abs{N_1^-}+\abs{N_1^+})} \Bigl( \frac{1+\delta e^{6J}}{1+\delta e^{-6J}} \Bigr)^{k(\abs{N_0^-}+\abs{N_0^+})}.
\end{equation}
Then:
\begin{align}\label{loggamma}
\begin{aligned}
\frac{\log{\pi_{\Lambda}^G(f^k)}}
     {{\abs{\Lambda}} }
     \leq 
     & -\frac{2 \log{Z_G}}{\abs{\Lambda}}+3J+2\lambda + k \log \left[ 1+\delta e^{-6J-2\lambda} \right] \\ 
     & +\frac{1}{\abs{\Lambda}} \log{\sum_{\Gamma} \left[ e^{-2J\abs{\Gamma}} \xi^k(\Gamma,0)\right]}.
\end{aligned}
\end{align}
The first three terms of the r.h.s.
of \eqref{loggamma} do not depend on 
$\delta$ and the fourth one is analytic
in $\delta$. Therefore, to prove that the
r.h.s. is analytic, it must be shown that
the last term is analytic.
As in \cite{pss}, 
$\Xi := \sum_{\Gamma} \left[ e^{-2J\abs{\Gamma}} \xi^k(\Gamma,0)\right]$
can be written as the partition function
of an abstract polymer gas. The analyticity 
of $\frac{\log \Xi}{|\Lambda|}$ follows by
showing that the activity of each polymer 
is sufficiently small. In our case a polymer $\gamma$
is a single 3d-Peierls' contour (defined above) and
its activity is 
$\rho(\gamma) = e^{-2J|\gamma|\xi^k(\gamma)},$
where $\xi^k(\gamma)$ is defined as in \eqref{xigamma}, when $\Gamma$ consists of the single contour $\gamma$. 
Then the proof can be concluded following the same steps of the proof of \cite[Lemma~2.2]{pss}.
This establishes the first part of 
Theorem~\ref{theorem1}.\\
The second part of theorem says: if $\lim_{\abs{\Lambda} \to \infty} \delta^2 \abs{\Lambda}=0$, then exist a $\tilde{J}$ such that, for $J<\tilde{J}$we have:
$
\lim_{\abs{\Lambda} \to \infty} \norma{\pi_{\Lambda}-\pi_{\Lambda}^G}_{TV}=0.
$
This follows applying \cite[Theorem~1.1]{dss12} and
noting that, for $J$ sufficiently small,
\begin{equation}
\sup _{x} \sum_{y} \tanh \left(2\left|J_{x, y}\right|\right) = 
\tanh \left(2\cdot6J\right)< 1.
\end{equation}

\textbf{Remark}: If some spin in $\Lambda$ is kept fixed, then, as already mentioned, the factor $2$ in the
last equality of \eqref{eq:from_conf_to_contours} becomes a $1$.
Consequentely, $\frac{2}{Z_G}$ in \eqref{pigamma} and $\frac{2 \log Z_G}{|\Lambda|}$ in \eqref{loggamma} become, respectively, $\frac{1}{Z_G}$ and $\frac{1\log Z_G}{|\Lambda|}$. The analiticity of
$\frac{1}{\abs{\Lambda}} \log\sum_{\Gamma} \left[ e^{-2J\abs{\Gamma}} \xi^k(\Gamma,0) \right]$
still holds since the set of Peierls' countours $\Gamma$ when some spin is fixed is a strict subset of the set of Peierls' contours when periodic boundary conditions are taken into account.

% \subsection{Proof of Theorem \ref{theorem2}}
\subsection{Proof of theorem~\ref{thm:low_temp_regime}}
We want to show that the mean value of a spin,
in the low temperature regime, at the centre of the lattice depends on the boundary conditions in the case of finite volume and continues to depend on the boundary even in the limit of infinite volume. We interpret this as the fact that the system is in the ordered phase.
Of course if we have a external magnetic field different from zero, all spins follow the orientation of such external field.\\
So to study the spontaneous behavior of the system, we go back to Hamiltonian \eqref{equation1}, \eqref{equation2} and set $\lambda = 0$. Moreover, from now on, we fix the external spins of $\Lambda$ (that we denote by $\partial^{ext}\Lambda$) to assume the value $+1$ that 
is we impose $+1$ boundary conditions. We have
\begin{equation}
\label{hamiltonian_no_field}
H^+(\sigma,\tau)=-\sum_{x\in\Lambda} [J\sigma_x(\tau_x^u+\tau_x^r+\tau_x^f)+q\sigma_x\tau_x]=-\sum_{x\in\Lambda} [J\tau_x(\sigma_x^d+\sigma_x^l+\sigma_x^b)+q\tau_x\sigma_x].
\end{equation}\\
From Proposition~\ref{proposition1} it follows:
\begin{equation}
\label{equation69}
H^+(\sigma,\sigma)=H^+(\sigma)-q|\Lambda^+|,
\end{equation}\\
with
\begin{equation}
\label{equation70}
H^+(\sigma)=-\sum_{(x,y)\in B^+_{\Lambda}} J\sigma_x\sigma_y,
\end{equation}\\
where $B^+_{\Lambda}$ is the set of all nearest neighbors pairs in  $\Lambda \cup \partial^{ext}\Lambda$:
\begin{equation}
\label{B+}
B_\Lambda^+=\{(x,y): x,y\in \Lambda \cup \partial^{ext}\Lambda, |x-y|=1\}.
\end{equation}
From now on, for convenience, we will omit the over-script $+$ in the Hamiltonian, that is we write $H$ in place of $H^+$.

From Lemma~\ref{lemma1}, we have:
\begin{equation}
\label{Zetaboundary}
\overrightarrow{Z_{\sigma}}=e^{q\abs{\Lambda}}e^{-H(\sigma)} \prod_{x \in \Lambda} \left(1+\delta e^{-2Jh_x(\sigma)\sigma_x}\right),
\end{equation}
where:
\begin{equation}
\label{Zetaboundary2}
\delta=e^{-2q} \text{ \ \ and \ \ } h_x(\sigma)=\sigma_x^d+\sigma_x^l+\sigma_x^b.
\end{equation}

We can, therefore, compute the mean value of $\sigma_0$.\\
We assume that the lattice goes from $-L/2$ to $+L/2$ in all the three directions, so $\sigma_0$ is the spin at the centre of the lattice, that is the furthermost from the boundary.\\
The mean value of $\sigma_0$ with positive boundary conditions is:
\begin{equation}
\label{mean}
\left\langle \sigma_0 \right\rangle_{\Lambda}^+ = \sum_{\sigma} \sigma_0 \pi_{\Lambda}(\sigma) = \sum_{\sigma} \frac{\sigma_0\overrightarrow{Z_{\sigma}}}{Z} = \frac{1}{Z}\sum_{\sigma, \tau} \sigma_0 e^{-H(\sigma, \tau)} = \frac{\sum_{\sigma, \tau} \sigma_0 e^{-H(\sigma, \tau)}}{\sum_{\sigma, \tau} e^{-H(\sigma, \tau)}}.
\end{equation}
Using \eqref{Zetaboundary} the last expression can be written as follow:

\begin{equation}
\label{mean2}
\left\langle \sigma_0 \right\rangle_{\Lambda}^+= \frac{\sum_{\sigma}\sigma_0e^{-H(\sigma)}\prod_{x \in \Lambda} \left(1+\delta e^{-2Jh_x(\sigma)\sigma_x}\right)}{\sum_{\sigma}e^{-H(\sigma)}\prod_{x \in \Lambda} \left(1+\delta e^{-2Jh_x(\sigma)\sigma_x}\right)}.
\end{equation}

Adapting the notation of the previous sections to the present one, we have:

\begin{equation}
\omega^G(\sigma) = e^{-H(\sigma)},
\end{equation}
\begin{equation}
f(\sigma) = \prod_{x \in \Lambda} \left( 1+\delta e^{-2Jh_x(\sigma)\sigma_x} \right),
\end{equation}
\begin{equation}
\overrightarrow{Z_{\sigma}}=e^{q\abs{\Lambda}}\omega^G(\sigma)f(\sigma),
\end{equation}
and, consequently,
\begin{equation}
Z=\sum_{\sigma}e^{q\abs{\Lambda}}\omega^G(\sigma)f(\sigma).
\end{equation}

Note that the expressions of $\pi_{\Lambda}^G(\sigma),\pi_{\Lambda}(\sigma), \pi_{\Lambda}^G(f)$ remain unchanged.

We now compute $\left\langle \sigma_0 \right\rangle_{\Lambda}^+$ in the 
low temperature regime $J \gg 1$.

Clearly:
\begin{equation}
\label{sigmaprob}
\left\langle \sigma_0 \right\rangle_{\Lambda}^+ = (+1)\PplusBC(\sigma_0=+1)+(-1)\PplusBC(\sigma_0=-1)=1-2\PplusBC(\sigma_0=-1).
\end{equation}

We can now estimate $\PplusBC(\sigma_0=-1)$ using a contour representation, defining the 3d-Peierls contours as in the previous section.\\
Let $\Gamma(\sigma)=\{\gamma_1,\dots,\gamma_N\}$ be the family of 3d-Peierls contours associated to the spin configuration $\sigma$.

If $\sigma_0=-1$ in a given configuration $\sigma$, then 
there exists at least one polyhedron in $\Gamma(\sigma)$ that surrounds $\sigma_0$.

We use the notation $\gamma_0 \odot \{0\}$ to denote a polyhedron $\gamma_0$ that surrounds $\sigma_0$.\\
Moreover, given a particular Peierls contour $\gamma_0$, we denote with $A_{\gamma_0}$ the set of family of contours containing $\gamma_0$:
\begin{equation}
\label{Agamma}
A_{\gamma_0} = \{\Gamma: \gamma_0 \in \Gamma\}.
\end{equation}

Then:
\begin{equation}\label{prob1}
\PplusBC(\sigma_0=-1) \leq \PplusBC\left(\cup_{\gamma_0 \odot \{0\}}A_{\gamma_0} \right) \leq \sum_{\gamma_0 \odot \{0\}}\PplusBC(A_{\gamma_0}),
\end{equation}
and
\begin{equation}
\label{prob2}
\PplusBC(A_{\gamma_0})=\frac{\sum_{\Gamma:\gamma_0\in\Gamma} \frac{\overrightarrow{Z_{\sigma}}}{Z}}{\sum_{\Gamma} \frac{\overrightarrow{Z_{\sigma}}}{Z}}=\frac{\sum_{\Gamma:\gamma_0\in\Gamma} \omega^G(\Gamma(\sigma))f(\Gamma(\sigma))}{\sum_{\Gamma} \omega^G(\Gamma(\sigma))f(\Gamma(\sigma))}.
\end{equation}

Where $\omega^G(\Gamma(\sigma))$ and $f(\Gamma(\sigma))$ are $\omega^G(\sigma)$ and $f(\sigma)$ written in terms of Peierls contours.

We now see how $\omega^G(\sigma)$ and $f(\sigma)$ 
can be written 
in terms of Peierls contours.

\begin{lemma}
\label{omegagamma}
\begin{equation}
\label{omegagammaeq}
\omega^G(\Gamma(\sigma))=e^{3JL^2(L+1)}e^{-2J|\Gamma|},
\end{equation}
where $\abs{\Gamma}$ is the total length of all Peierls contours in $\Gamma(\sigma)$:
\begin{equation}
\label{absgamma}
|\Gamma|=\sum_i |\gamma_i|.
\end{equation}
\end{lemma}
\begin{proof}
Let $B^+$ be the number of bonds in of positive sign $B^+_{\Lambda}$  (the number of nearest neighbors with same sign) and with $B^-$ the number of bonds of negative sign (the number of nearest neighbors with opposite sign). Then
\begin{equation}
    \omega^G(\Gamma(\sigma))=e^{-H(\sigma)}=e^{\sum_{(x,y)\in B_{\Lambda}^+} J\sigma_x\sigma_y}=e^{J(B^{+}-B^{-})}.
\end{equation}

Now, let $ \{ \gamma_1, \gamma_2, \dots, \gamma_N\}$ be a contours configuration associate to the spin configuration $\sigma$. Then, by construction:
\begin{equation}
\abs{\Gamma}=\sum_{i=1}^N \abs{\gamma_i} = B^-
\end{equation}
as above.
Moreover: $B^{+}+B^{-}=B_{TOT}$.

Then:
\begin{equation}
B^{+}=B_{TOT}-B^{-}=B_{TOT}-|\Gamma|,
\end{equation}
and
so:
\begin{equation}
\omega^G(\Gamma(\sigma))=e^{J(B^{+}-B^{-})}=e^{J(B_{TOT}-2|\Gamma|)}=e^{JB_{TOT}}e^{-2J|\Gamma|}.
\end{equation}
Finally, is easy to show that $B_{TOT}=3L^2(L+1).$
Then, we have:
\begin{equation}
\omega^G(\Gamma(\sigma))=e^{3JL^2(L+1)}e^{-2J|\Gamma|}.
\end{equation}
\end{proof}

Lemma~\ref{omegagamma} describes how $\omega^G(\sigma)$
can be written in terms of Peierls contours. 
Similarly, also $f(\sigma)$
can be written
in terms of Peierls contours. To this end, recalling \eqref{eqf}, 
\eqref{xi}, \eqref{xigamma}), we have:
\begin{equation}
\label{fgamma}
f(\Gamma)=\left(1+\delta e^{-6J}\right)^{|\Lambda|}\left(\frac{1+\delta e^{-2J}}{1+\delta e^{-6J}}\right)^{|N_2|}\left(\frac{1+\delta e^{2J}}{1+\delta e^{-6J}}\right)^{|N_1|}\left(\frac{1+\delta e^{6J}}{1+\delta e^{-6J}}\right)^{|N_0|},
\end{equation}
with:
\begin{itemize}[label=\adfbullet{43}]
\item $N_3=\{x\in\Lambda: \sigma_x^d+\sigma_x^l+\sigma_x^b=3\sigma_x\}=\{x\in\Lambda: h_x(\sigma)=3\sigma_x\}.$
\item $N_2=\{x\in\Lambda: \sigma_x^d+\sigma_x^l+\sigma_x^b=\sigma_x\}=\{x\in\Lambda: h_x(\sigma)=\sigma_x\}.$
\item $N_1=\{x\in\Lambda: \sigma_x^d+\sigma_x^l+\sigma_x^b=-\sigma_x\}=\{x\in\Lambda: h_x(\sigma)=-\sigma_x\}.$
\item $N_0=\{x\in\Lambda: \sigma_x^d+\sigma_x^l+\sigma_x^b=-3\sigma_x\}=\{x\in\Lambda: h_x(\sigma)=-3\sigma_x\}.$
\end{itemize}

We can now compute $\PplusBC(\sigma_0=-1)$ using equations (\ref{prob1}), (\ref{prob2}), (\ref{omegagammaeq}) and (\ref{fgamma})
and noting that the terms 
$e^{3JL^2(L+1)}$ and
$\left(1+\delta e^{-6J}\right)^{|\Lambda|}$
appear both in the numerator and in the denominator as factors for all
$\Gamma$:

\begin{equation}
\begin{aligned}
    & \PplusBC(\sigma_0=-1)
    \leq \sum_{\gamma_0 \odot \{0\}}\frac{\sum_{\Gamma:\gamma_0\in\Gamma} \omega^G(\Gamma(\sigma))f(\Gamma(\sigma))}{\sum_{\Gamma} \omega^G(\Gamma(\sigma))f(\Gamma(\sigma))}\\ 
        & = \sum_{\gamma_0 \odot \{0\}}
            \frac{
                \sum_{\Gamma:\gamma_0\in\Gamma} 
                    e^{-2J\sum_{i=0}^N |\gamma_i|} 
                    \left(\frac{1+\delta e^{-2J}}{1+\delta e^{-6J}}\right)^{|N_2|}
                    \left(\frac{1+\delta e^{2J}}{1+\delta e^{-6J}}\right)^{|N_1|}
                    \left(\frac{1+\delta e^{6J}}{1+\delta e^{-6J}}\right)^{|N_0|}
        }   {
                \sum_{\Gamma} 
                    e^{-2J\sum_{i=1}^N |\gamma_i|} 
                    \left(\frac{1+\delta e^{-2J}}{1+\delta e^{-6J}}\right)^{|N_2|}
                    \left(\frac{1+\delta e^{2J}}{1+\delta e^{-6J}}\right)^{|N_1|}
                    \left(\frac{1+\delta e^{6J}}{1+\delta e^{-6J}}\right)^{|N_0|}
        }.
\end{aligned}
\end{equation}

Since the denominator contains more terms than the numerator,
it is a standard task (Peierls' argument) to show that
\begin{equation}
    \PplusBC(\sigma_0=-1) 
        \leq \sum_{\mathclap{\gamma_0 \odot \{0\}}}
            e^{-2J|\gamma_0|} 
            \left(\frac{1+\delta e^{-2J}}{1+\delta e^{-6J}}\right)^{|N_2(\gamma_0)|}
            \left(\frac{1+\delta e^{2J}}{1+\delta e^{-6J}}\right)^{|N_1(\gamma_0)|}
            \left(\frac{1+\delta e^{6J}}{1+\delta e^{-6J}}\right)^{|N_0(\gamma_0)|}.
\end{equation}

Since $1+\delta e^{-6J}\geq 1$, then 
$\frac{1+\delta e^{-2J}}{1+\delta e^{-6J}}<1+\delta e^{-2J}$,
$\frac{1+\delta e^{2J}}{1+\delta e^{-6J}}<1+\delta e^{2J}$ 
and
$\frac{1+\delta e^{6J}}{1+\delta e^{-6J}}<1+\delta e^{6J}$.
Therefore
\begin{equation}
    \PplusBC(\sigma_0=-1)\leq 
        \sum_{\gamma_0 \odot \{0\}} e^{-2J|\gamma_0|} 
        \left(1+\delta e^{-2J}\right)^{|N_2(\gamma_0)|}
        \left(1+\delta e^{2J}\right)^{|N_1(\gamma_0)|}
        \left(1+\delta e^{6J}\right)^{|N_0(\gamma_0)|}.
\end{equation} 
    
Now observe that
\begin{equation}
\label{obsgamma}
|\gamma_0|=3|N_0(\gamma_0)|+2|N_1(\gamma_0)|+1|N_2(\gamma_0)|.
\end{equation}
Indeed each site in $N_0(\gamma_0)$ contributes to the length of $\gamma_0$ with $3$ unitary plates; each site in $N_1(\gamma_0)$ contributes with two plates and each site in $N_2(\gamma_0)$ contributes with one plate.\\
Then, a simple algebraic computation leads to:
\begin{equation}\label{prob5eq}
    \PplusBC(\sigma_0=-1)
    \leq \sum_{\gamma_0 \odot \{0\}} 
        \left(e^{-2J}+\delta e^{-4J}\right)^{|N_2(\gamma_0)|}
        \left(e^{-4J}+\delta e^{-2J}\right)^{|N_1(\gamma_0)|}
        \left(e^{-6J}+\delta\right)^{|N_0(\gamma_0)|}.
\end{equation}

Let us proceed now by transforming the sum on the contours of Peierls $\gamma_0 \odot \{0\}$ in a sum over their lengths: $|\gamma_0|=k.$\\
We indicate with $\eta_0(k)$ the number of Peierls contours of length $k$ that surrounds the site $0$. Moreover $k\geq 6$ for closed contours.\\
Then, we can write:
\begin{equation}
\label{prob6eq}
\PplusBC(\sigma_0=-1) \leq \sum_{k=6}^{\infty} \eta_0(k)D^k,
\end{equation}
with
\begin{equation}\label{D}
    D=\max\set*{
        \left(e^{-2J}+\delta e^{-4J}\right);\,
        \left(e^{-4J}+\delta e^{-2J}\right)^{1/2};\,
        \left(e^{-6J}+\delta\right)^{1/3}
        },
\end{equation}
indeed, thanks to (\ref{obsgamma}), each factor of the form $\left(e^{-2J}+\delta e^{-4J}\right)$ contributes with exponent $1$ to the length of $\gamma_0$, each factor of the form $\left(e^{-4J}+\delta e^{-2J}\right)$ contributes with exponent $1/2$ and each factor of the form $\left(e^{-6J}+\delta\right)$ contributes with exponent $1/3.$

For $J$ large enough (low temperature regime):
\begin{align}
    D 
        & =\max\set*{
            \left(e^{-2J}+\delta e^{-4J}\right);\,
            \left(e^{-4J}+\delta e^{-2J}\right)^{1/2};\,
            \left(e^{-6J}+\delta\right)^{1/3}
            }
        \label{D2}
        = \left(e^{-6J}+\delta\right)^{1/3}.
\end{align}

Finally, an estimate of the number of contours of length $k$ surrounding the origin is given by Ruelle's lemma: $\eta_0(k)\leq 3^k$.

Then (\ref{prob6eq}) becomes:
\begin{equation}
\label{prob7eq}
\PplusBC(\sigma_0=-1) \leq \sum_{k=6}^{\infty} 3^k\left(e^{-6J}+\delta\right)^{k/3}=\sum_{k=6}^{\infty} \left[3\left(e^{-6J}+e^{-2q}\right)^{1/3}\right]^k.
\end{equation}

This expression identifies a geometric series with common ratio $r=3\left(e^{-6J}+e^{-2q}\right)^{1/3}$.

The series is convergent if $r<1$, that is: $\left(e^{-6J}+e^{-2q}\right)^{1/3}<\frac{1}{3}$.

Under this conditions, the series converges to the value: $\sum_{k=6}^{\infty} r^k=\frac{r^6}{1-r}$.

Then (\ref{prob7eq}) becomes: $\PplusBC(\sigma_0=-1) \leq \frac{3^6\left(e^{-6J}+e^{-2q}\right)^2}{1-3\left(e^{-6J}+e^{-2q}\right)^{1/3}}$.

This expression tends toward zero for $J\gg 1$ and $q\gg 1$ uniformely in $\Lambda$.\\
Therefore, for $J$ and $q$ large enough:
\begin{equation}
\label{eqP+}
\PplusBC(\sigma_0=-1) < 1/2.
\end{equation}
This happens if the following condition holds:
\begin{equation}
\label{condtrans}
\frac{3^6\left(e^{-6J}+e^{-2q}\right)^2}{1-3\left(e^{-6J}+e^{-2q}\right)^{1/3}}<\frac{1}{2}.
\end{equation}

Finally, if (\ref{condtrans}) holds, recalling (\ref{sigmaprob}) and (\ref{eqP+}), we have:
\begin{equation}
\label{sigma+}
\left\langle \sigma_0 \right\rangle_{\Lambda}^+ >0.
\end{equation}

Moreover, in the same manner, we can show that, under the condition (\ref{condtrans}) we have:
\begin{equation}\label{eqP-}
    \PminusBC(\sigma_0=-1) > 1/2,
\end{equation}
and then
\begin{equation}\label{sigma-}
    \left\langle \sigma_0 \right\rangle_{\Lambda}^- <0.
\end{equation}

From (\ref{sigma+}) and (\ref{sigma-}), we have:
\begin{equation}\label{sigmadiffer}
    \left\langle \sigma_0 \right\rangle_{\Lambda}^+ \neq
    \left\langle \sigma_0 \right\rangle_{\Lambda}^-.
\end{equation}

And this inequality holds even in the thermodynamic
 limit $\Lambda \to \infty$:
\begin{equation}
\label{sigmadifferinfty}
\left\langle \sigma_0 \right\rangle^+ \neq \left\langle \sigma_0 \right\rangle^-.
\end{equation}

So, the 1-point correlation functions are not unique, in the low-temperature regime, with respect to boundary conditions. This suffice to asserts that the equilibrium state, that is the family of all $n$-points correlation functions in the thermodynamic limit (with $n=1,2,...,\abs{\Lambda}$), is not unique at low temperature, but it depends on boundary conditions.

\subsection{Proof of Theorem~\ref{thm:high_temp_regime}}
\begin{proof}
We start by rewriting the Hamiltonian \eqref{hamiltonian_no_field}:
\begin{equation}
\label{H_low_temp}
H(\sigma,\tau)=-\sum_{x\in\Lambda} [J\tau_x(\sigma_x^d+\sigma_x^l+\sigma_x^b)+q\tau_x\sigma_x]=-\sum_{e \in E_{\altBox}} J_e \sigma_{e_1} \tau_{e_2},
\end{equation}
where $E_{\altBox}=E_J \cup E_q$ is the edge set of the finite tetrahedral lattice $\altBox$ with periodic boundary
conditions and $e_1,e_2$ are two sites in $\altBox$ linked by the edge $e$.
Moreover:
\begin{equation}
J_e=
\begin{cases}
J \text{ \ \ if \ \ } e \in E_J\\
q \text{ \ \ if \ \ } e \in E_q.
\end{cases}
\end{equation}
Setting $\vec{\sigma} = (\sigma,\tau) \in \chi^2 \text{ and } b_e=\sigma_{e_1} \tau_{e_2} \in \{-1,1\}$, we can write:
\begin{equation}
\label{H_good}
H(\sigma,\tau)=H(\vec{\sigma})=-\sum_{e \in E_{\altBox}} J_e \sigma_{e_1} \tau_{e_2}=-\sum_{e \in E_{\altBox}} J_e b_e.
\end{equation}
Thus, the partition function is:
\begin{equation}
\label{partifuncthigh}
Z_{\altBox} = \sum_{\vec{\sigma}} e^{-H(\vec{\sigma})} =\sum_{\vec{\sigma}} e^{\sum_{e \in E_{\altBox}} J_e b_e} = \sum_{\vec{\sigma}} \prod_{e \in E_{\altBox}} e^{J_e b_e}.
\end{equation}
Write $
e^{J_e b_e}  
= \cosh(J_e) [ 1 + b_e\tanh(J_e)].
$
Then, the partition function can be rewritten as:
\begin{equation}
\label{partifuncthigh2}
Z_{\altBox} = [\cosh(J_e)]^{\abs{E_{\altBox}}} \sum_{\vec{\sigma}} \prod_{e \in E_{\altBox}} [ 1 + b_e\tanh(J_e)].
\end{equation}
Developing the product $\prod_{e \in E_{\altBox}} [ 1 + b_e\tanh(J_e)]$, we get terms of the type:
\begin{equation}
[\tanh(J_e)]^N b_{e_1} b_{e_2} \dots b_{e_N},
\end{equation}
which has a clear geometric interpretation: the set of bonds $b_{e_1} \dots b_{e_N}$ form a graph (connected or not) in $\altBox$ whose links are nearest neighbors. Performing the sum over $\vec{\sigma}$ we get that the only graphs which yield a non vanishing contribution to
$\sum_{\vec{\sigma}} b_{e_1} \dots b_{e_N}$, and hence to the partition function, are those
whose vertices have incidence number two or four, while all other graphs are zero once the sum over configurations $\vec{\sigma}$ has been done. Graphs of this type are called non-vanishing.\\
If the graph $b_{e_1} \dots b_{e_N}$ is non-vanishing, then:
\begin{equation}
    \sum_{\vec{\sigma}} b_{e_1} \dots b_{e_N} = 2^{\abs{\altBox}}.
\end{equation}
We can naturally split a non vanishing graph into non intersecting connected components which we will call lattice animals. A lattice animal $\gamma$ is thus nothing but a graph $g$ with edge set $E_{\gamma}= \{ e_1 \dots e_N \} \in E_{\altBox}$ formed by nearest neighbor links. The allowed lattice animals
are only those $\gamma$ with incidence number at the vertices equal to two or four. We denote with $\mathbb{L}_{\altBox}$ the set of all possible lattice animals in $\altBox$.\\
Two lattice animals $\gamma$ and $\gamma'$ are non overlapping (i.e. compatible), and we write $\gamma \sim \gamma'$ if and only if $\gamma \cap \gamma' = \emptyset$. We will denote shortly $\abs{\gamma} = \abs{E_{\gamma}}$. 
In conclusion we can write (from \ref{partifuncthigh2}):
\begin{equation}
\label{partufuncthigh3}
Z_{\altBox} = [\cosh(J_e)]^{\abs{E_{\altBox}}} 2^{\abs{\altBox}} \Xi_{\altBox}(J_e),
\end{equation}
where
\begin{equation}
\label{partpolim}
\Xi_{\altBox}(J_e) = 1 + \sum_{n \geq 1} \sum_{\substack{ \{ \gamma_1, \dots, \gamma_n \} \subset \mathbb{L}_{\altBox}, \\ \gamma_i \sim \gamma_j }} \xi(\gamma_1) \dots \xi(\gamma_n)
\end{equation}
is the partition function of a hard core polymer gas in which polymers are lattice animals, i.e. elements of $\mathbb{L}_{\altBox}$ with the incompatibility relation $\gamma \nsim \gamma'$.
Each polymer $\gamma$ has an activity given by:
\begin{equation}
\label{activitymk}
\xi(\gamma) = [\tanh(J)]^{2k}[\tanh(q)]^{2m},
\end{equation}
where $2k$ and $2m$ denote respectively the number of $J-$edges and the number of $q-$edges of the closed polymer $\gamma$, so $2k + 2m = \abs{\gamma}$. Actually, we can imagine the lattice $\altBox$ as a collection of layer of hexagonal lattice, the sites of which are linked only by $J-$edges, and these layers are linked to each other by edges of type $q$. So, a closed circuit must have an even number of $q-$segments and an even number of $J-$segments and the number of $q-$type segments must be less then or equal to half the number of $J-$type segments.
To control the analyticity of $Z_{\altBox}$, we can apply the Fernandez-Procacci convergence criterion (see \cite{procaccifernandez}) to the $\Xi_{\altBox}(J_e)$.\\
Namely, we need to find numbers $\mu(\gamma) \in (0,+\infty)$ such that:
\begin{equation}
\label{FP1}
\xi(\gamma) \leq \frac{\mu(\gamma)}{\Xi_{\mathbb{L}_{\gamma}}(\vec{\mu})},
\end{equation}
where
$
\mathbb{L}_{\gamma} = \{ \gamma' \in \mathbb{L}_{\altBox} : \gamma' \nsim \gamma \}
$
is the set of all polymers in $\altBox$ incompatible with $\gamma$ (i.e. the set of all polymers that intersect $\gamma$).\\
Setting $\mu(\gamma) = \xi(\gamma) e^{a\abs{\gamma}}$, the condition (\ref{FP1}) becomes
\begin{equation}
\label{FP2}
\Xi_{\mathbb{L}_{\gamma}}(\vec{\mu}) \leq e^{a\abs{\gamma}},
\end{equation}
where
\begin{align}
\label{Xibig}
    \Xi_{\mathbb{L}_{\gamma}}(\vec{\mu}) 
    & = 1 + \sum_{n=1}^{\abs{\gamma}} \frac{1}{n!} \sum_{\substack{(\gamma_1, \dots, \gamma_n) \subset \mathbb{L}_{\gamma}^n, \\ \gamma_i \sim \gamma_j }} \abs{\mu(\gamma_1) \dots \mu(\gamma_n)}\\
    & = 1 + \sum_{n=1}^{\abs{\gamma}} \frac{1}{n!} \sum_{\substack{(\gamma_1, \dots, \gamma_n) \subset \mathbb{L}_{\gamma}^n, \\ \gamma_i \sim \gamma_j }} \prod_{i=1}^n \abs{\xi(\gamma_i)} e^{a\abs{\gamma_i}},
\end{align}
with the sum $\sum_{\substack{(\gamma_1, \dots, \gamma_n) \subset \mathbb{L}_{\gamma}^n, \\ \gamma_i \sim \gamma_j }} (\cdot)$ running over all ordered $n$-tuples of polymers.\\
Consider now the factor:
\begin{equation}
\sum_{\substack{(\gamma_1, \dots, \gamma_n) \subset \mathbb{L}_{\gamma}^n, \\ \gamma_i \sim \gamma_j }} \prod_{i=1}^n \abs{\xi(\gamma_i)} e^{a\abs{\gamma_i}},
\end{equation}
we have thus to choose $n$ lattice animals $(\gamma_1, \dots, \gamma_n)$ all incompatible with a given lattice animal $\gamma$ and all pairwise compatible. We recall that two lattice animals are incompatible if they share a vertex in $\altBox$ . Since $\gamma$ has $|\gamma|$ edges, we can find at most $|\gamma|$ animals incompatibles with $\gamma$ and pairwise compatible. Thus, the factor above is zero whenever $ n > |\gamma|, \gamma_i \nsim \gamma$. \\
We want to rearrange the sum over all ordered $n$-tuples of polymers $(\gamma_1, \dots, \gamma_n) \subset \mathbb{L}_{\gamma}^n$ in a sum over all polymers $\{ \gamma_1, \dots, \gamma_n \} \subset \mathbb{L}_{\altBox}$, regardless of their order. Recalling that the factor above is zero whenever $n > |\gamma|$, we have:
\begin{align}
\begin{aligned}
    & \sum_{\substack{(\gamma_1, \dots, \gamma_n) \subset 
          \mathbb{L}_{\gamma}^n, \\ \gamma_i \sim \gamma_j }
    } \prod_{i=1}^n \abs{\xi(\gamma_i)} e^{a\abs{\gamma_i}}\\ 
    & = (|\gamma|)(|\gamma|-1)\dots(|\gamma|-n+1) \sum_{\gamma_1 \in \mathbb{L}_{\altBox}} \abs{\xi(\gamma_1)} e^{a\abs{\gamma_1}}
    % \sum_{\gamma_2 \in \mathbb{L}_{\altBox}} \abs{\xi(\gamma_2)} e^{a\abs{\gamma_2}}
    \dots \sum_{\gamma_n \in \mathbb{L}_{\altBox}} \abs{\xi(\gamma_n)} e^{a\abs{\gamma_n}} \\
    & \leq (|\gamma|)(|\gamma|-1)\dots(|\gamma|-n+1) 
    \parens[\bigg]{\sup_{x \in \altBox} \sum_{x \in \gamma \in \mathbb{L}_{\altBox}} \abs{\xi(\gamma)} e^{a\abs{\gamma}} 
    }^n\\
    & = \binom{|\gamma|}{n} n! 
    \parens[\bigg]{ \sup_{x \in \altBox} \sum_{x \in \gamma \in \mathbb{L}_{\altBox}} \abs{\xi(\gamma)} e^{a\abs{\gamma}} 
    }^n.
\end{aligned}
\end{align}
    
Then, (\ref{Xibig}) becomes:
\begin{align}
\begin{aligned}
\label{Xibig2}
    \Xi_{\mathbb{L}_{\gamma}}(\vec{\mu}) 
    & = 1 + \sum_{n=1}^{\abs{\gamma}} {\frac{1}{n!}} \binom{|\gamma|}{n} {n!} 
    \parens[\bigg]{ \sup_{x \in \altBox} \sum_{x \in \gamma \in \mathbb{L}_{\altBox}} \abs{\xi(\gamma)} e^{a\abs{\gamma}}
    }^n\\
    & =  \parens[\bigg]{ 1 + \sup_{x \in \altBox} \sum_{x \in \gamma \in \mathbb{L}_{\altBox}} \abs{\xi(\gamma)} e^{a\abs{\gamma}} 
    }^{|\gamma|}.
\end{aligned}    
\end{align}
The convergence condition (\ref{FP2}) becomes:
\begin{equation}
\label{FP3}
 \parens[\bigg]{ 1 + \sup_{x \in \altBox} 
            \sum_{x \in \gamma \in \mathbb{L}_{\altBox}} 
            \abs{\xi(\gamma)} e^{a\abs{\gamma}} 
 }^{|\gamma|} \leq e^{a\abs{\gamma}}.
\end{equation}
and hence
\begin{equation}
\label{FP4}
\sup_{x \in \altBox} \sum_{x \in \gamma \in \mathbb{L}_{\altBox}} \abs{\xi(\gamma)} e^{a\abs{\gamma}} \leq e^a - 1.
\end{equation}
Observe finally that, due to the structure of the lattice, the function
\begin{equation}
f(x) = \sum_{x \in \gamma \in \mathbb{L}_{\altBox}} \abs{\xi(\gamma)} e^{a\abs{\gamma}}
\end{equation}
is constant as $x$ varies in $\altBox$. Therefore (\ref{FP4}) is equivalent to the condition
\begin{equation}
\sum_{0 \in \gamma \in \mathbb{L}_{\altBox}} \abs{\xi(\gamma)} e^{a\abs{\gamma}} \leq e^a - 1.
\end{equation}
where 0 is the ``origin'' in $\altBox$.
Now, recalling (\ref{activitymk}), the above condition becomes:
\begin{equation}
\label{FP5}
\sum_{0 \in \gamma \in \mathbb{L}_{\altBox}} \abs*{[\tanh(J)]^{2k}[\tanh(q)]^{2m} e^{a(2k+2m)}} \leq e^a - 1.
\end{equation}
We want to convert the sum over $\gamma$'s passing through $0$ in a sum over their lengths $\abs{\gamma}=2k+2m$. To this end, we observe that in a closed circuit, the number $q-$segments must be less then or equal to the number of $J-$segments: $2m \leq 2k$; and the minimal number of edges must be $6$: $2k+2m \geq 6$.
So, condition (\ref{FP5}) becomes:
\begin{equation}
\label{FP6}
\sum_{k \geq 2} \sum_{\substack{ m = 0 \\ k+m \geq 3 }}^{k} \abs[\Big]{[\tanh(J)]^{2k}[\tanh(q)]^{2m} e^{a(2k+2m)} \sum_{\substack{0 \in \gamma \in \mathbb{L}_{\altBox} \\  \abs{\gamma}=2k+2m}} 1 }\leq e^a - 1.
\end{equation}
The sum:
\begin{equation}
\sum_{\substack{0 \in \gamma \in \mathbb{L}_{\altBox}  \\  \abs{\gamma}=2k+2m}} 1
\end{equation}
corresponds to the number of closed circuits of length $2k+2m$ passing through $0$.\\
We can find this number imagining to start from a certain point and doing $2k$ steps of $J-$type and $2m$ steps of $q-$type, until returning to the starting point. As long as we move on a layer we carry out all $J-$type steps, while when we change layer we carry out a $q-$type step. As long as we move on a layer, we have $2^{2k-2m}$ possible circuits (at each node, we have $2$ possible directions); however, occasionally we have to insert a change of layer (a segment of type $q$), for a total of $2m$ segments of this type. So we have to insert $2m$ step of type $q$ among the $2k$ steps of type $J$. We can do this in $\binom{2k}{2m}$ ways. Once we change layers, we have $3$ possible directions for the first step in this new layer, so we need to consider also a factor $3^{2m}$.
So, the total number of such circuits is:
\begin{equation}
\sum_{\substack{0 \in \gamma \in \mathbb{L}_{\altBox} \\  \abs{\gamma}=2k+2m}} 1 = 2^{2k-2m} 3^{2m} \binom{2k}{2m}.
\end{equation}
Hence, condition (\ref{FP6}) becomes:
\begin{equation}
\label{FP7}
\sum_{k \geq 2} \sum_{\substack{ m = 0 \\ k+m \geq 3 }}^{k} \abs*{[\tanh(J)]^{2k}[\tanh(q)]^{2m} e^{a(2k+2m)} 2^{2k} \left(\frac{3}{2}\right)^{2m} \binom{2k}{2m}} \leq e^a - 1.
\end{equation}
yielding
\begin{equation}
\label{FP8}
\sum_{\substack{ k \geq 2 \\ k+m \geq 3 }} \abs*{[2\tanh(J) e^a]^{2k} \sum_{m = 0}^{k} \binom{2k}{2m} \left[\left(\tfrac{3}{2}\right)\tanh(q) e^a\right]^{2m} }\leq e^a - 1.
\end{equation}
We observe that the two sums must satisfy the constraint $k+m \geq 3$, that is the close circuits condition. We can then extend these sums to all value $k \geq 2$, $q \geq 0$ without any constrain if we subtract by hand the only term forbidden by the constraint; this term corresponds to $k=2$ and $m=0$. So, we finally have:
\begin{equation}
\label{FP9}
\sum_{k \geq 2} \abs*{ [2\tanh(J) e^a]^{2k} \sum_{m =0}^{k} \binom{2k}{2m} \left[\left(\tfrac{3}{2}\right)\tanh(q) e^a\right]^{2m} - [2\tanh(J) e^a]^4} \leq e^a - 1.
\end{equation}
It is a standard task to show that:
\begin{equation}
\label{sumpari2}
\sum_{m=0}^k \binom{2k}{2m} x^{2m} = \frac{1}{2} \left( (x-1)^{2k} + (x+1)^{2k} \right),
\end{equation}
with $x = \left(\frac{3}{2}\right)\tanh(q) e^a$.
Further, setting $y=2\tanh{(J)}e^a$,  condition (\ref{FP9}) becomes:
\begin{equation}
\label{FP10}
\frac{1}{2} \sum_{k \geq 2} \abs*{ y^{2k} \left[ \left( x-1 \right)^{2k} + \left(x+1 \right)^{2k} \right] - y^4} \leq e^a - 1.
\end{equation}

We can finally perform the remaining sums over $k$:
\begin{equation}
\sum_{k \geq 2} [y(x\pm1)]^{2k}=
\sum_{k \geq 2} [y^2(x\pm1)^2]^k =
\frac{[y^2(x\pm1)^2]^2}{1-y^2(x\pm1)^2} = \frac{y^4(x\pm1)^4}{1-y^2(x\pm1)^2}
\end{equation}
yielding
\begin{equation}
\label{FP12}
\abs*{\frac{1}{2} \left[ \frac{y^4(x-1)^4}{1-y^2(x-1)^2} + \frac{y^4(x+1)^4}{1-y^2(x+1)^2} \right] -y^4} \leq e^a - 1.
\end{equation}
Finally, recalling the form of $x$ and $y$, we have:
\begin{tiny}
\begin{equation}\label{FP13}
    \abs*{\frac{1}{2} \left[ \frac{(2\tanh(J) e^a)^4(\frac{3}{2}\tanh(q) e^a-1)^4}{1-(2\tanh(J) e^a)^2(\frac{3}{2}\tanh(q) e^a-1)^2} + \frac{(2\tanh(J) e^a)^4(\frac{3}{2}\tanh(q) e^a+1)^4}{1-(2\tanh(J) e^a)^2(\frac{3}{2}\tanh(q) e^a+1)^2} \right] +
    (2\tanh(J) e^a)^4}\leq e^a - 1.
\end{equation}
\end{tiny}

This expression is the condition that $J$ and $q$ must satisfy to be sure that \linebreak
\mbox{$f_{\altBox}(J,q)=\frac{1}{\abs{\altBox}} \ln Z_{\altBox}(J,q)$} is an analytic function for in $J$ and $q$.
Numerical evaluations show that a good value of $a$ is $a=0.15$. For this value, the expression (\ref{FP13}) identifies the region 
below the lower curve in Fig.\ref{twocurves}.
Hence, for values of $J$ and $q$ small enough (i.e. in the aforementioned region), $
f_{\altBox}(J,q)$ is analytic.\\
\end{proof}

Thus, we have shown that in low-temperature regime the system is in the ordered phase, while in the high-temperature regime it is in the disordered one. Therefore there must be a critical line in $J-q$ plane that separates the ordered phase from the disordered one and this curve must lie in the region between the two curves as shown in Figure~\ref{twocurves}.
This fact is well supported by numerical simulations.

\begin{figure}
\centering
\includegraphics[width=0.8\textwidth]{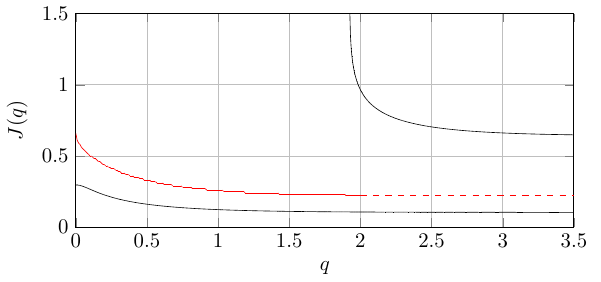}
\caption{The curves delimiting the Low Temperature and High Temperature regions identified in equations \eqref{condtrans} and \eqref{FP13} together with a numerical approximation (in red) of the
conjectured critical curve.}
\label{twocurves}
\end{figure}

\section{Numerical simulations}\label{sec:numerical_simulations}
In this section we present our numerical results (obtained with techniques similar to those used in \cite{dauparallel})
concerning the critical curve in the $(q,J)$ plane and 
discuss the behavior of the critical exponents
of the magnetic susceptibility
as $q$ varies.

To this end we consider the alternate dynamics taking
place on a class of ``tetrahedral'' non homogeneous lattices.
It is possible to give a geometric interpretation to the pair interaction
thinking to $J$ and $q$ to be proportional to the inverse of the
distance of the lattice points. In this way, the three dimensional lattice
can be thought to be a collection of honeycomb layers (where the 
the side of each hexagon is $\frac{1}{J}$) at distance
$\frac{1}{q}$ one from another.
With this picture in mind, we observe that when $q=J$ the dynamics
lives on the diamond lattice whereas for $q \to \infty$ the lattice
becomes the simple cubic one. Moving in the other direction (towards
smaller values of $q$) the interaction between the layers
becomes weaker and weaker up to the point ($q = 0$) when they become
independent so that the system resembles a collection of well separated graphene sheets.

In this framework, estimating the critical curve in the $(q,J)$ plane, 
amounts to finding the critical ``size'' $J$ of the hexagons for each
``distance'' $q$ between the sheets.
Moreover, it is reasonable to think that the transition from the three dimensional model to a collection of independent two dimensional
ones 
implies that the critical exponents of the magnetic susceptibility undergoes a sharp change at $q=0$.

Both the critical values of $J$ and the critical exponents can be
estimated by looking at the variance of the magnetization of the system.
Indeed, at the critical $J$ this variance diverges.

To estimate the variance of the magnetization, we considered its
sample variance computed over a long run for a somewhat large system, see below for more details. 
The critical values of $J$ that we found are
consistent with the actual critical value
of $J$ for the Ising model on the honeycomb lattice ($q=0$) and
with recent numerical estimates for the critical $J$ for the
diamond lattice ($q=J$) and the simple cubic lattice ($q$ large). This fact gives us some confidence
on our findings 
concerning the whole critical curve.

Our results are summarized in Figs.~\ref{twocurves},
\ref{fig:critical_temperatures}
and~\ref{fig:critical_curve}. In particular
Fig.~\ref{twocurves} shows that the estimated
critical curve lies in the region between
the curves of equations~\eqref{condtrans} and~\eqref{FP13}
delimiting the low and the high temperature region respectively.
Fig.~\ref{fig:critical_temperatures} shows the values of the 
normalized standard deviation of the magnetization as a 
function of $J$ for $q=0$, $q=J$ and $q=2$ corresponding, respectively,
to the
collection of two dimensional honeycomb lattices, the diamond lattice and,
ideally, the simple cubic lattice. Our estimate of the critical $J$ for each value of  $q$,
denoted by $\hat{J}_c(q)$, is given
by the value at which the variance is maximal.
We have $\hat{J}_c(0) = 0.659$. In this case the analytical 
critical value is $J_c^\mathrm{h.l.} \approx 0.659$ (see \cite{shaken2d})
For $q = J$ we obtained $\hat{J}_c(J) = 0.370$ whereas in
\cite{diamond} the numerical estimate is $J_c^\mathrm{d.l.} = 0.370$. Finally,
setting $q = 2$ we estimated 
$\hat{J}_c(2) = 0.226$.
In \cite{3d} $J_c^\mathrm{s.c.} = 0.222$

\begin{figure}
\centering
\includegraphics[width=0.7\textwidth]{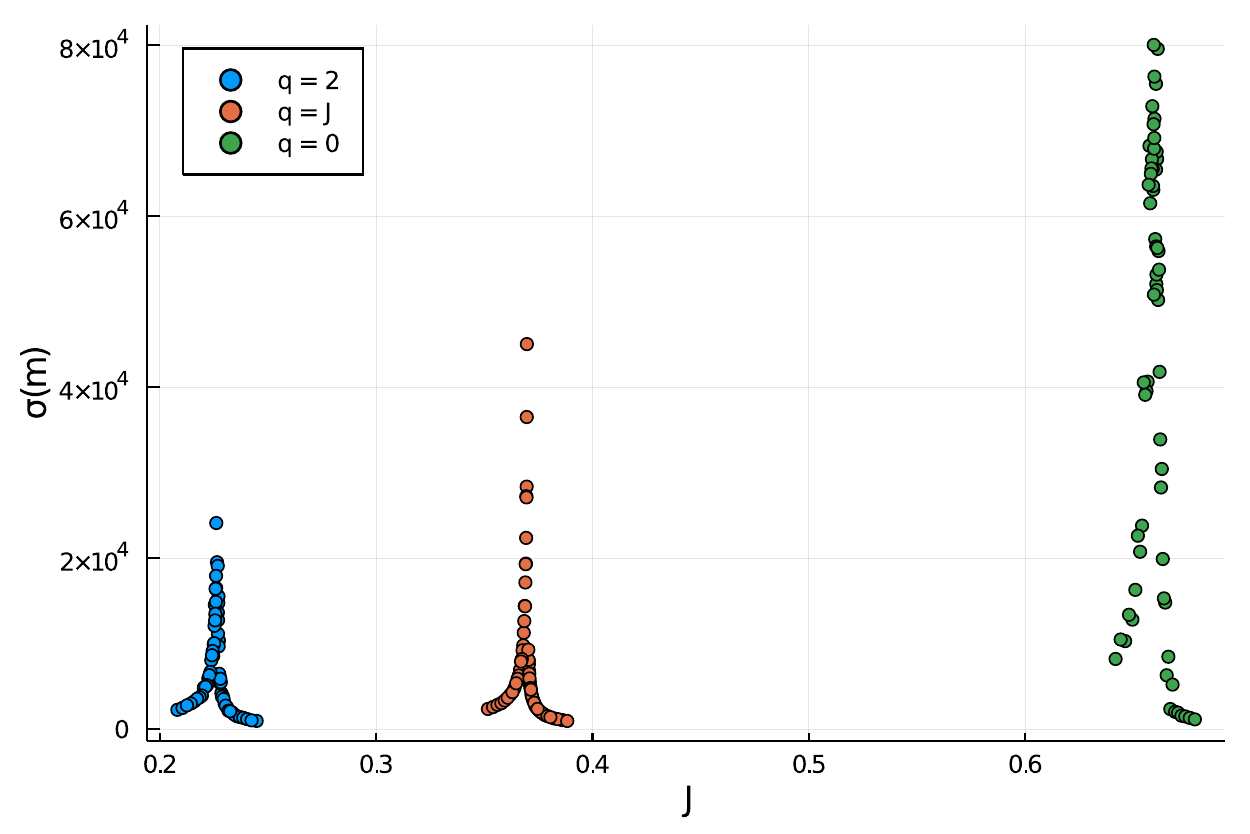}
\caption{The standard deviation of the magnetization as a function of $J$ for $q = 2$ (ideally the limit $q \to \infty$), $q = J$ and $q = 0$. Our estimates for the critical temperature are $J_c = 0.226$ for the cubic lattice, $J_c = 0.370$ for the tetrahedral diamond lattice and $J_c = 0.669$ for the 2d honeycomb lattice.}
\label{fig:critical_temperatures}
\end{figure}

\begin{figure}
\centering
    \includegraphics[width=0.45\textwidth]{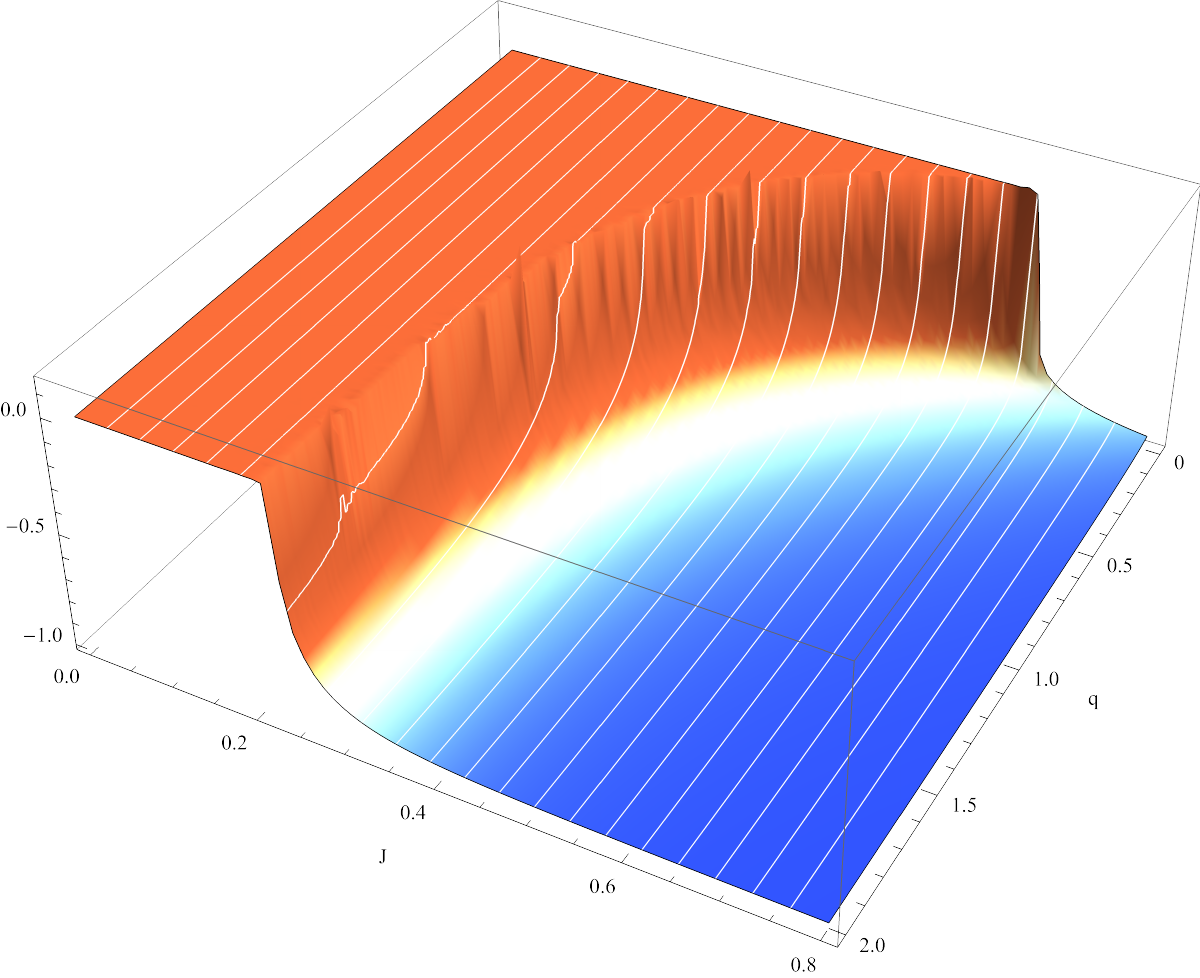}
    \quad
    \includegraphics[width=0.45\textwidth]{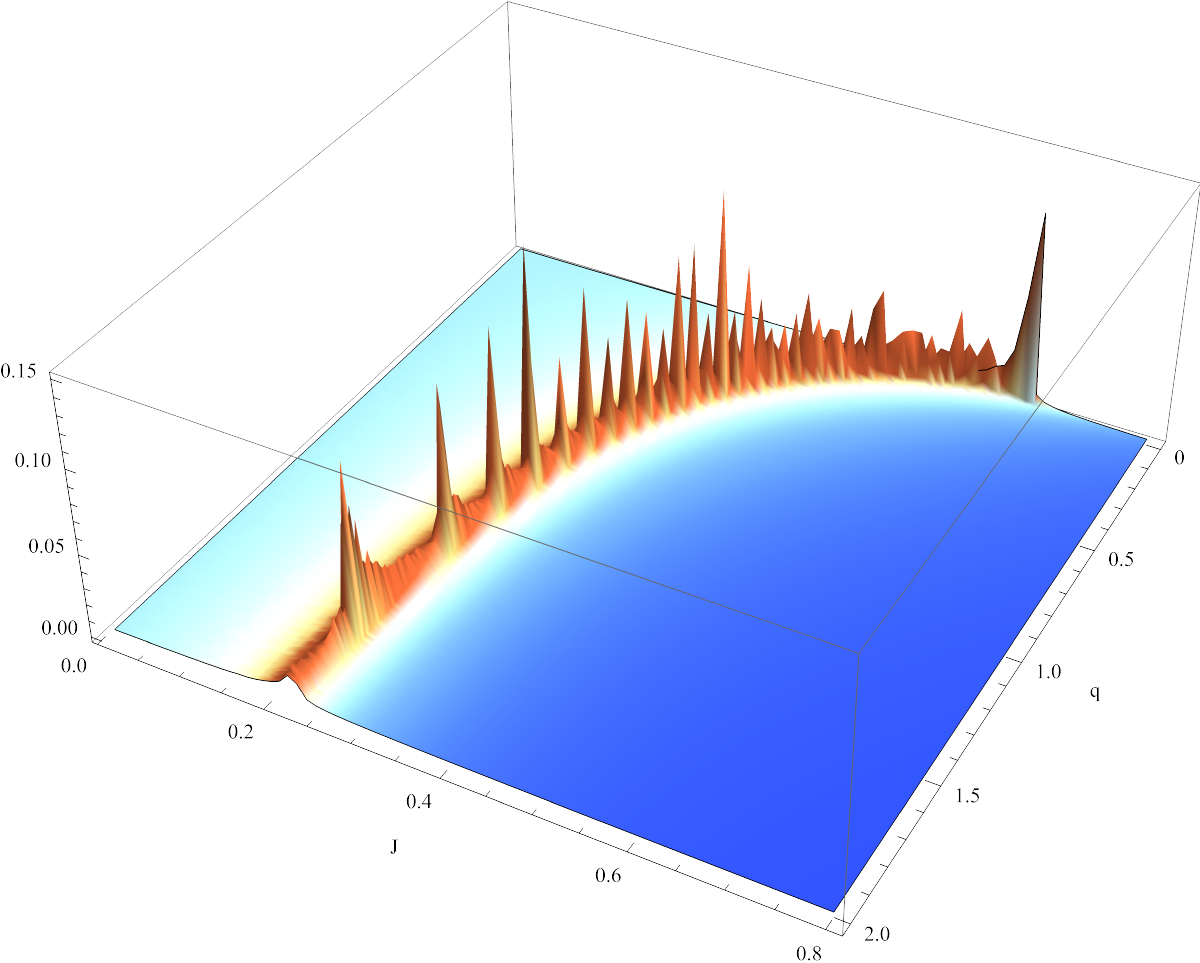}
\caption{The ``critical behavior'' of the mean magnetization (left) and the standard deviation (right) for the tested values of $J$ and $q$}
\label{fig:critical_curve}
\end{figure}

%%%%%%%%%%%%%%%%%%%%%%%%%%%%%
%HERE DISCUSSION ON CRITICAL EXPONENTS
Approaching the critical temperature, the magnetic susceptibility
$\chi$
(the variance of the magnetization) diverges with a power law
with some \emph{critical exponent} $\gamma$
(see, e.g., \cite{macritical, Thompson})
\begin{align}
    \chi = \left(
           \frac{\abs{T - T_c}}{T_c}
           \right)^{-\gamma}.
\end{align}

Recalling that, in this paper we wrote
$J$ in place of the usual $\beta J$, 
where $\beta = \frac{1}{k_b T}$,
we get
\begin{align}
    \gamma = -\frac{\log(\chi)}
                   {\log(\abs{J - Jc})}
             + C
\end{align}
for a suitable constant $C$.

Note that the value of $\gamma$ is 
related to the dimension
of the system and it is the same for the whole class of Ising-like
lattice systems with the same dimension (see \cite{macritical}).
We estimated $\gamma$ for several values of $q$ ranging from
$0$ to $2$, that is for geometries ranging from a collection
of $2d$ honeycomb lattices to the simple cubic lattice.
In this case, our estimates are not meant to 
determine the values of the critical exponent with high accuracy.
Rather, as long as our values are consistent with those available
in the literature, we want use them
to support our conjecture that the systems
retains a three dimensional structure for all positive $q$s.

Our results are summarized in Fig.~\ref{fig:slope_critical_exponents}
and Fig.~\ref{fig:critical_coefficients}.
There it is possible to see that, for $q > 0$ 
both the high and low temperature critical exponents are quite
close to the value $\gamma \approx 1.237$ that is
the critical value for three dimensional Ising systems
(see \cite{25orderGamma}). 
On the other hand, 
as soon as $q=0$, our estimate jumps to a value that
is much closer to the critical value for two dimensional Ising
system ($\gamma=\frac{7}{4}$, see \cite{Thompson}).

These findings show that our model is able to capture
through the variation of the parameter $q$ the
dimensional transition
in the geometry of the system.

\begin{figure}
    \centering
    \includegraphics[width=0.4\textwidth]
                    {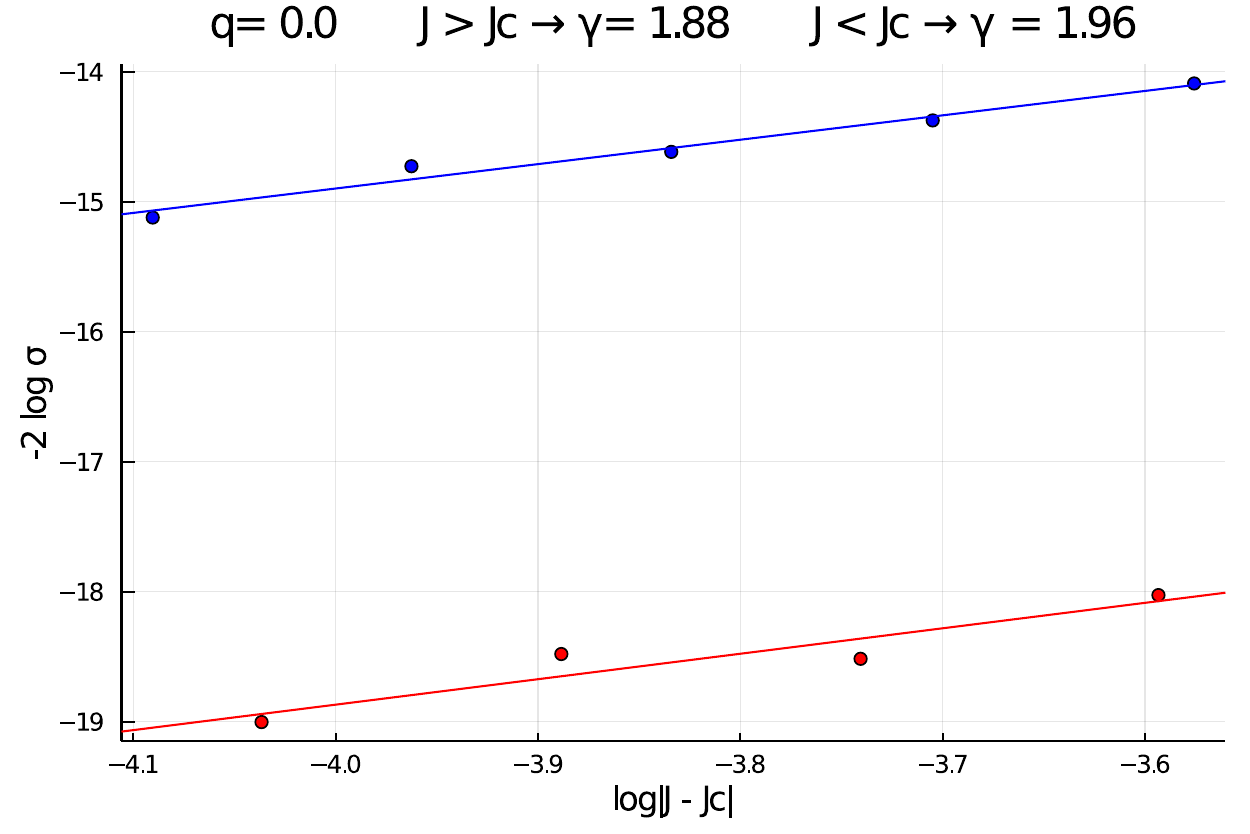}
                    \hspace{0.1\textwidth}
    \includegraphics[width=0.4\textwidth]
                    {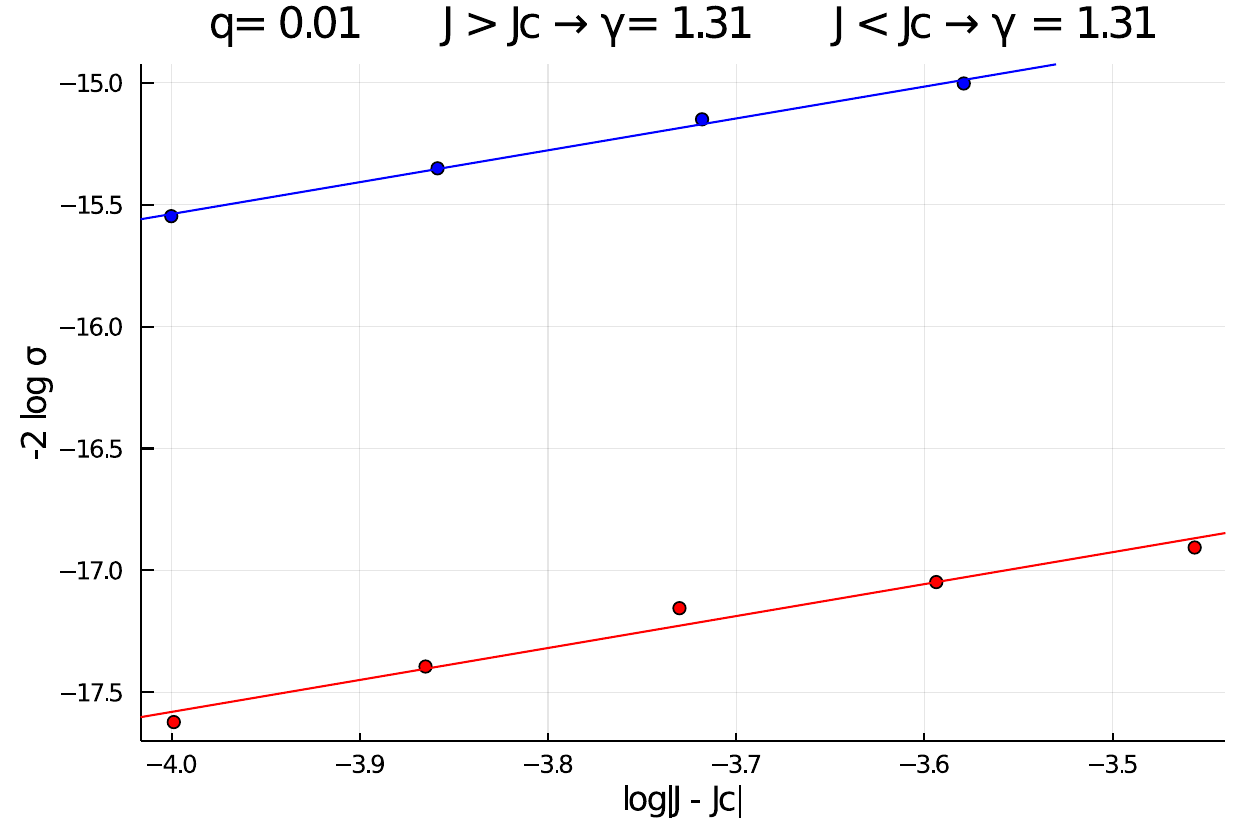}\\
    
    \vspace{0.5cm}
                    
    \includegraphics[width=0.4\textwidth]
                    {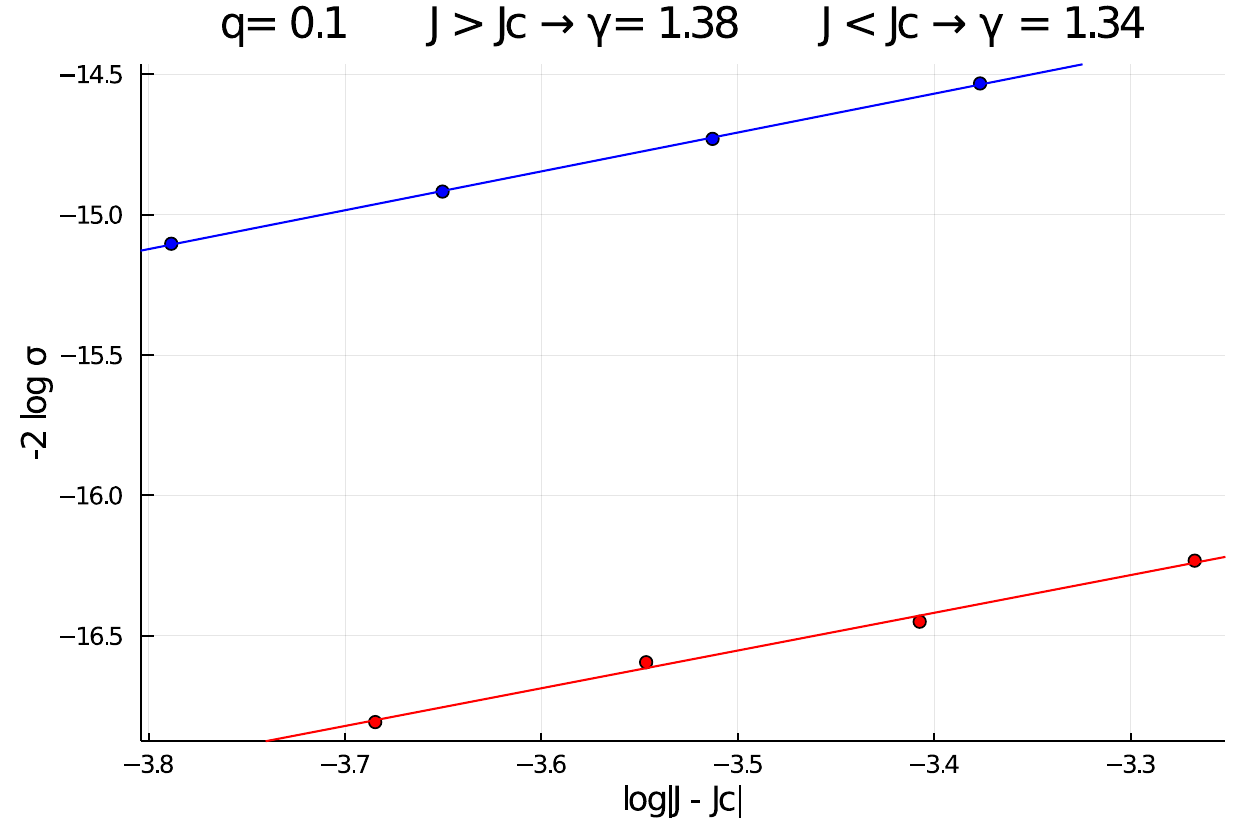}
                    \hspace{0.1\textwidth}
    \includegraphics[width=0.4\textwidth]{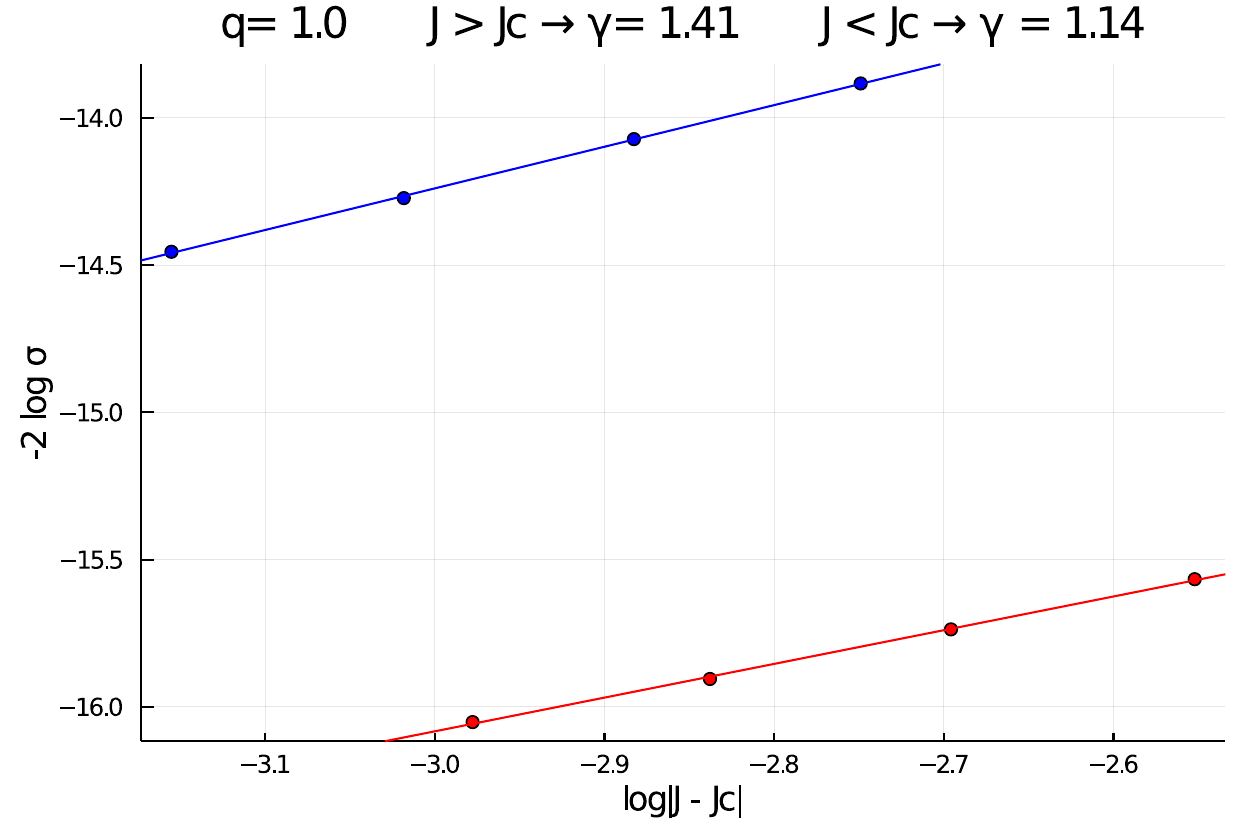}
    \caption{The regression lines for $\gamma$ for some values of
            $q$. In each chart the red line is obtained looking at the
            values $J < J_c$ (high temperature) whreas the blue line
            is obtained looking at the values $J > J_c$ (low temperature).}
    \label{fig:slope_critical_exponents}
\end{figure}

%%%%%%%%%%%%%%%%%%%%%%%%%

\begin{figure}
    \centering
    \includegraphics[width=0.4\textwidth]
                    {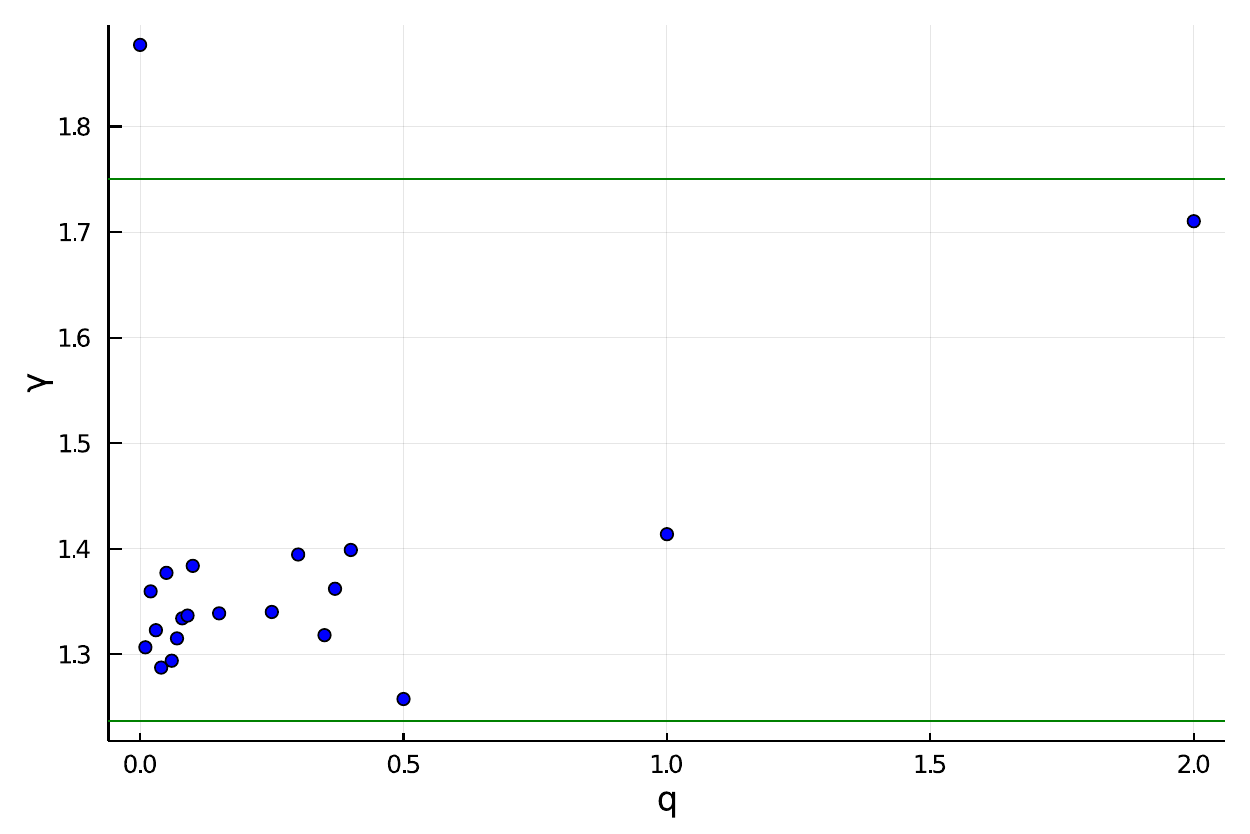}
                    \hspace{0.1\textwidth}
    \includegraphics[width=0.4\textwidth]
                    {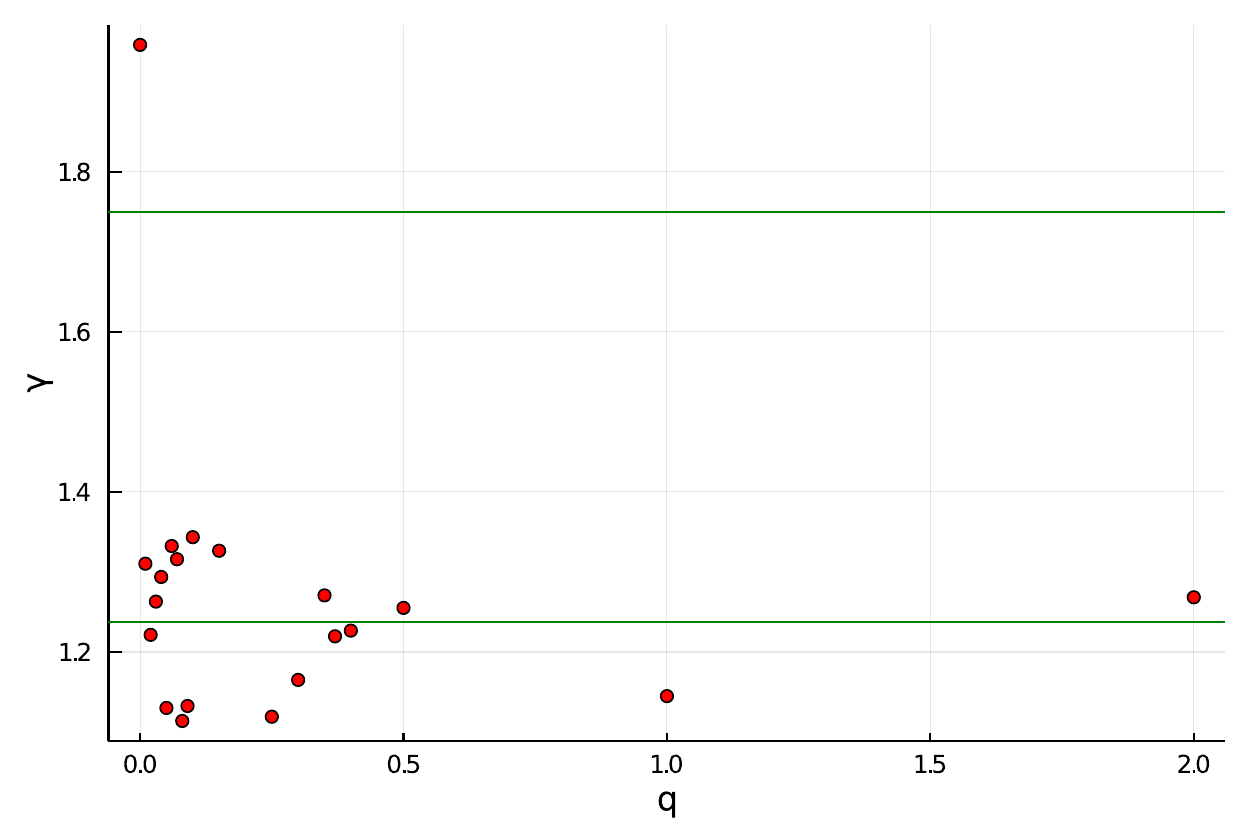}
    \caption{The low temperature ($J > J_c$, in blue) and 
            high temperature ($J < J_c$, in red) coefficients
            for values of $q$ between $0$ and $2.0$. The
            solid lines correspond to the ``true'' values of
            the critical exponents for the two and three dimensional systems. The numerically determined values appear to be
            close to the critical value of the three dimensional system for $q>0$ and to the critical value of the two dimensional system for $q=0$. The behavior of numerical estimate for the
            low temperature critical exponent in the case $q=2$ is
            likely to be due to the limited size of the simulated system.
    }
    \label{fig:critical_coefficients}
\end{figure}

%%%%%%%%%%%%%%%%%%%%%%%%%%%%%%

\subsection{Numerical details and heuristic discussion}\label{sec:numerical_details}

We simulated the shaken dynamics, 
without external magnetic field, on a 
$96 \times 96 \times 96$ grid on which we
imposed periodic boundary conditions.

We considered a grid of points in the $(q,J)$
and for each point in the grid, we started the simulation from the configuration with all spin set to $-1$ and let the system perform 510000 steps of the shaken dynamics (that is, 1020000 half steps). We considered the first 10000 steps as a ``transient'' and collected statistics on the final 500000 steps. In particular, for each pair of parameters $(q,J)$ we computed the the average and variance (over time) of the magnetization.

Simulations have been carried over using the language ``julia''.

A heuristic insight on why this procedure should
be useful it is possible to argue as follows.
Letting the dynamics start from the configuration with all spins taking value $-1$, it is expected to reach very rapidly a local minimizer 
of the free energy  and start visiting configurations that are close to this minimizer. 

In the high temperature regime, the minimizers of the free energy are expected to have all zero mean magnetization and if the parameters $(q,J)$ are in the high temperature region, the dynamics is likely to return very quickly to a state where the number of plus and minus spins is essentially the same. As a consequence, it is possible to conjecture that the average (over time) of the magnetization is very close to zero and that its variance is very small (see figure \ref{fig:critical_curve}).

In the low temperature region, the free energy has minimizers whose mean magnetization is closer (and closer as the system freezes) to $\pm 1$. The colder the system, the higher the free energy barriers separating the ``positive magnetization'' minimizers from the ``negative magnetization'' ones. As the chain evolves, the dynamics will overcome a free energy barrier of magnitude $\Delta$ with a probability that is exponentially small in $\Delta$. 
Therefore, since the system starts from the configuration where all spins are $-1$, it will very likely reach the vicinity of one of these $-1$ minimizer and will stay, with very high probability, in the region where the minimizers of the free energy have negative mean magnetization. The typical time to observe a transition to the $+1$ minimizers are exponentially large in the volume, and hence they are way beyond the possibility of a numerical simulation. With probability very close to $1$ the system will remain captured by the $-1$ minimizers. Consequently, also in the low temperature region we can expect a very small variance for the average magnetization whereas its mean is likely to be more and more negative as the system becomes colder (see figure \ref{fig:critical_curve}).

Around the critical temperature, the free energy has minimizers with both positive and negative mean magnetization. However, the ``valleys'' of the free energy landscape where these minimizers sit are rather shallow and, therefore, the dynamics is expected to move between minimizers whose mean magnetization has opposite signs. An evolution of this type will produce an average magnetization that is close to zero. 
Nevertheless, the variance of the magnetization is expected, in this case, to increase when the temperature approaches its critical value. Note that the general theory of critical phenomena (see again \cite{macritical}) shows that it is rather delicate to measure the features of the systems close to the critical temperature. The relatively good results we obtained with the simulations presented above show that the shaken (or alternate) dynamics is able to capture the features of the system also when the parameters are close-to-critical.

\section*{Acknowledgements}
BS acknowledges the support of the Italian MIUR Department of Excellence grant (CUP E83C18000100006).
AT acknowledges the support of the H2020 Project Stable and Chaotic Motions in the Planetary Problem
(Grant 677793 StableChaoticPlanetM of the European Research Council).

%%%%%%%%%%%%%%%%%%%%%%%%%%%%%%%%%%%%%%%%%%%%%%%%%%%%%%%%%%%%%%%%%%%%%%%%%%%%%%%%

% \cleardoublepage
\addcontentsline{toc}{chapter}{\bibname}

\providecommand{\bysame}{\leavevmode\hbox to3em{\hrulefill}\thinspace}
\providecommand{\MR}{\relax\ifhmode\unskip\space\fi MR }
\providecommand{\MRhref}[2]{%
  \href{http://www.ams.org/mathscinet-getitem?mr=#1}{#2}
}
\providecommand{\href}[2]{#2}

%%%%%%%%%%%%%%%%%%%%%%%%%%%%%%%%%%%%%%%%%%%%%%%%%%%%%%%%%%%%%%%%%%%%%%%%%%%%%

\end{document}